\pdfoutput=1
\newcommand{\fullversion}{}
\ifdefined\fullversion
\documentclass[a4paper,UKenglish]{lipics-v2018}
\else
\newcommand{\shortversion}{}
\documentclass[a4paper,UKenglish]{lipics-v2018}
\fi

\usepackage[utf8]{inputenc}
\usepackage{microtype}
\usepackage{epsfig}
\usepackage{epstopdf}
\usepackage{graphicx}
\DeclareGraphicsExtensions{.pdf,.png,.jpg,.eps}
\usepackage{amsthm}
\usepackage{latexsym}
\usepackage{amsfonts}
\usepackage{amssymb}
\usepackage{amsmath}
\usepackage{bm}
\usepackage{pseudocode}
\usepackage{array}
\usepackage{color}
\usepackage{multirow}
\usepackage{url}
\usepackage{caption} 
\usepackage{tikz}
\usepackage{makecell}
\usepackage{algorithm}
\usepackage{algpseudocode}

\algtext*{EndFunction}
\algtext*{EndIf}
\newcolumntype{C}[1]{>{\centering\let\newline\\\arraybackslash\hspace{0pt}}m{#1}}

\newcommand{\MS}[1]{\ensuremath{\Delta}}
\newcommand{\set}[1]{\ensuremath{\mathcal{S}_{#1}}}

\newcommand{\eps}{\ensuremath{\varepsilon}}

\newcommand{\predec}{{\textsf{pred}}}
\newcommand{\select}{{\textsf{select}}}
\newcommand{\rank}{{\textsf{rank}}}

\newcommand{\rankA}{{\textsf{drankA}}}
\newcommand{\vrankA}{{\textsf{rankA}}}
\newcommand{\vselectA}{{\textsf{selectA}}}
\newcommand{\dselectA}{{\textsf{dselectA}}}

\newcommand{\suffixsum}{{\textsf{ss}}}
\newcommand{\iss}{{\textsf{iss}}}
\newcommand{\dsuffixA}{{\textsf{ssA}}}

\newcommand{\vissA}{{\textsf{issA}}}

\newcommand{\ssA}{{\textsf{ssA}}}
\newcommand{\issA}{{\textsf{issA}}}

\newcommand{\summ}{{\textsf{sum}}}
\newcommand{\search}{{\textsf{search}}}
\newcommand{\pred}{{\textsf{pred}}}

\newcommand{\theTitle}{Approximate Query Processing over Static Sets and Sliding Windows}
\title{\theTitle}
\titlerunning{\theTitle}
\author{Ran Ben Basat}{Harvard University, Cambridge, USA}{ran@seas.harvard.edu}{}{Supported by the Zuckerman Foundation, the Technion Hiroshi Fujiwara cyber security research center, and the Israel Cyber Directorate.}

\author{Seungbum Jo}{University of Siegen, Germany}{Seungbum.Jo@uni-siegen.de}{ https://orcid.org/0000-0002-8644-3691}{The author of this paper is supported by the DFG research project
	LO748/11-1}

\author{Srinivasa Rao Satti}{Seoul National University, South Korea}{ssrao@cse.snu.ac.kr}
{https://orcid.org/0000-0003-0636-9880}{}

\author{Shubham Ugare}{IIT Guwahati, Guwahati, India}{ugare.dipak@iitg.ac.in}{}{}

\authorrunning{R.\,B. Basat, S. Jo, S.\,R. Satti and S. Ugare}

\Copyright{Ran Ben-Basat, Seungbum Jo, Srinivasa Rao Satti and Shubham Ugare}

\subjclass{Theory of computation$\rightarrow$Data compression}


\category{}

\ifdefined\shortversion
\EventEditors{Wen-Lian Hsu, Der-Tsai Lee, and Chung-Shou Liao}
\EventNoEds{3}
\EventLongTitle{29th International Symposium on Algorithms and Computation (ISAAC 2018)}
\EventShortTitle{ISAAC 2018}
\EventAcronym{ISAAC}
\EventYear{2018}
\EventDate{December 16--19, 2018}
\EventLocation{Jiaoxi, Yilan, Taiwan}
\EventLogo{}
\SeriesVolume{123}
\ArticleNo{31} 
\else
\nolinenumbers
\fi

\newtheorem*{theorem*}{Theorem}
\newtheorem*{lemma*}{Lemma}

\newcommand{\newInputLetter}{\ensuremath{{\rho}}}

\newcommand{\UBnBlocks}{\ceil{{n/\nu}+1}}

\newcommand{\lastBlockEnd}{\ensuremath{g}}

\newcommand{\lnrLBSymbol}{\mathcal B_{\ell,n,\lnrErrSymbol}}

\newcommand{\minDelta}{1}
\newcommand{\lnrErrSymbol}{\delta}
\newcommand{\sensitivity}{\ensuremath{\widetilde{\lnrErrSymbol}}}

\newcommand{\ceil}[1]{ \left\lceil{#1}\right\rceil}
\newcommand{\floor}[1]{ \left\lfloor{#1}\right\rfloor}

\newcommand{\parentheses}[1]{ \left({#1}\right)}

\newcommand{\biggParentheses}[1]{ \bigg({#1}\bigg)}

\newcommand{\logp}[1]{\log\parentheses{#1}}
\newcommand{\lgp}[1]{\lg\parentheses{#1}}
\newcommand{\Omegap}[1]{\Omega\parentheses{#1}}
\newcommand{\Op}[1]{O\parentheses{#1}}
\newcommand{\Thetap}[1]{\Theta\parentheses{#1}}

\newcommand{\cdotpa}[1]{\cdot\parentheses{#1}}

\newcommand{\oneOverE}{ \eps^{-1} }

\newcommand{\range}[2][0]{#1,1,\ldots,#2}

\newcommand{\frange}[1]{\set{\range{#1}}}

\newcommand{\sFactor}{(1+o(1))}
\newcommand{\sNegFactor}{(1-o(1))}
\newcommand{\brackets}[1]{\left[#1\right]}

\newcommand{\nBits}{\mathfrak b}
\newcommand{\ranker}{\ensuremath{\mathbb A}}

\newcommand{\numBits}{\ensuremath{\mathit{numElems}}}
\newcommand{\setBits}{\ensuremath{\mathit{totalSum}}}
\newcommand{\lastBit}{\ensuremath{\mathit{oldest_\newInputLetter}}}
\newcommand{\outBits}{\ensuremath{\mathit{out}}}
\newcommand{\outBitsVal}{\ensuremath{\parentheses{\nu - \parentheses{(i-o)\mod \nu}}}}

\newcommand{\bs}{{\sc Basic-Summing}}

\newcommand{\add}  [1][] { {\sc Add}$(#1)$}

\newcommand{\window}{n}

\newcommand{\lgw}{\lg \window}

\newcommand{\blockOffset}{o}
\newcommand{\inputVariable}{x}

\newcommand{\bsrange}{ \ell }

\begin{document}
\maketitle
\begin{abstract}
Indexing of static and dynamic sets is fundamental to a large set of applications such as information retrieval and caching. Denoting the characteristic vector of the set by $B$, we consider the problem of encoding sets and multisets to support \emph{approximate} 
versions of the operations $\rank{}(i)$ (i.e., computing $\sum_{j \le i}B[j]$) 
and $\select{}(i)$ (i.e., finding $\min\{p\mid\rank{}(p)\ge i\}$) queries.
We study multiple types of approximations (allowing an error in the query or the result) and present lower bounds and succinct data structures for several variants of the problem.
We also extend our model to sliding windows, in which we process a stream of elements and compute \emph{suffix sums}. This is a generalization of the window summation problem that allows the user to specify the window size \emph{at query time}. Here, we provide an algorithm that supports updates and queries in constant time while requiring just $(1+o(1))$ factor more space than the fixed-window summation~algorithms. 
\end{abstract}
\section{Introduction}
\label{sec:intro}

Given a bit-string $B[1 \dots n]$ of size $n$, one of the fundamental and well-known problems proposed by
Jacobson~\cite{Jacobson:1988:SSD:915547}, is to construct a space-efficient data structure
which can answer $\rank{}$ and $\select{}$ queries on $B$ efficiently.
For $b \in \{0, 1\}$, these queries are defined as follows.

\begin{itemize}
\item{$\rank{}_{b}(i, B)$ :} returns the number of $b$'s in $B[1 \dots i]$.
\item{$\select{}_{b}(i, B)$ :} returns the position of the $i$-th $b$ in $B$.
\end{itemize}

A bit vector supporting a subset of these operations is one of the basic building blocks
in the design of various succinct data structures. Supporting these operations in constant time, with close to the optimal 
amount of space, both theoretically and practically, has received a wide range of attention~\cite{DBLP:journals/cj/JoJORS17, Munro:2001:SES:375519.375532, DBLP:conf/wea/NavarroP12, DBLP:conf/alenex/OkanoharaS07, DBLP:journals/talg/RamanRS07}.
Some of these results also explore trade-offs that allow more query time while reducing the space.

We also consider related problems in the streaming model,
where a quasi-infinite sequence of integers arrives, and our algorithms need to support the operation of appending a new item to the end of the stream. 
For $i\in\{1,\ldots, n\}$, let $S_i$ be the sum of the last $i$ integers. Here, $n$ is the maximal suffix size we support queries for.
For streaming, we consider processing a stream of elements, and answering two types of queries, \textit{suffix sum} ($\suffixsum{}$) and \textit{inverse suffix sum} ($\iss{}$), defined as:
\begin{itemize}
	\item {$\suffixsum{}(i, n)$:} returns $S_i$ for any $1 \le i \le n$.
	\item {$\iss{}(i, n)$:} returns the smallest $j$, $1 \le j \le n$, such that $\suffixsum{}(j, n) \ge i$.	
\end{itemize}

In this paper, our goal is to obtain space efficient data structures for supporting a few relaxations of 
these queries efficiently using an amount of space below the theoretical minimum (for the unrelaxed versions), ideally.
To this end, we define {\em approximate} versions of \rank{} and \select{} queries, and 
propose data structures for answering \textit{approximate \rank{} and \select{} queries} on multisets and bit-strings.
We consider the following approximate queries with an \textit{additive} error $\delta > 0$.

\begin{itemize}
\item {$\vrankA{}_{b}(i, B, \delta)$:} returns any value $r$ which satisfies $\rank{}_{b}(i-\delta, B) < r \le \rank{}_{b}(i, B)$. If $\rank{}_{b}(i-\delta, B) = \rank{}_{b}(i, B)$, then $\vrankA{}_{b}(i, B, \delta) = \rank{}_{b}(i, B)$.
\item {$\rankA{}_{b}(i, B, \delta)$:} returns any value $r$ which satisfies $\rank{}_{b}(i, B)-\delta < r \le \rank{}_{b}(i, B)$.
\item {$\vselectA{}_{b}(i, B, \delta)$:} returns any position $p$ which satisfies
$\select{}_{b}(i-\delta, B) < p \le \select{}_{b}(i, B)$.
\item {$\dselectA{}_{b}(i, B, \delta)$:} returns any position $p$ which satisfies 
$\select{}_{b}(i, B)-\delta < p \le \select{}_{b}(i, B)$.
\item {$\ssA{}(i,n, \delta)$:} returns any value $r$ which satisfies $\suffixsum{}(i, n) -\delta < r \le \suffixsum(i, n)$.
\item {$\issA{}(i, n,\delta)$:} returns any value $r$ which satisfies $\iss{}(i-\delta, n) < r \le \iss(i, n)$.
\end{itemize}

We propose data structures for supporting approximate \rank{} and \select{} queries on bit-strings efficiently. Our data structures uses less space than that is required to answer the exact queries and most of data structures use optimal space. 
We also propose a data structure for supporting $\ssA{}$ and $\issA{}$ queries on binary streams while supporting updates efficiently.
Finally, we extend some of these results to the case of larger alphabets. 
For all these results, we assume the 
standard word-RAM model~\cite{miltersen-survey} with word size 
$\Theta(\lg{n})$ if it is not explicitly mentioned.

\subsection{Previous work}
\noindent
{\bf Rank and Select over bit-strings. }
Given a bit-string $B$ of size $n$, it is clear that at least $n$ bits are necessary to support
$\rank{}$ and $\select{}$ queries on $B$. Jacobson~\cite{Jacobson:1988:SSD:915547}
proposed a data structure for answering $\rank{}$ queries on $B$ in constant time using $n+o(n)$ bits.
Clark and Munro~\cite{Clark:1996:EST:313852.314087}
extended it to support both $\rank{}$ and $\select{}$ queries in constant time with $n+o(n)$ bits.
For the case when there are $m$ $1$'s in $B$, at least 
$\mathcal{B}(n, m)$ bits\footnote{$\mathcal{B}(n, m) = \lg{\ceil{{n \choose m}}}$ bits
is the information-theoretic lower bound on space
for storing a subset of size $m \le n$ from the universe $\{1,2, \dots n\}$.} are necessary to support $\rank{}$ and $\select$ on $B$. Raman et al.~\cite{DBLP:journals/talg/RamanRS07} proposed a data structure that supports both operations in constant time while using $\mathcal{B}(n, m)+o(n)+O(\lg\lg{m})$ bits. 
Golynski et al.~\cite{DBLP:journals/algorithmica/GolynskiOR014} gave an asymptotically optimal time-space trade-off for supporting $\rank{}$ and $\select{}$ queries on $B$.
A slightly related problem of {\em approximate color counting} has been considered in El-Zein et al.~\cite{DBLP:conf/isaac/El-ZeinMN17}.
\\\\
\noindent
\ifdefined\fullversion
A natural generalization of the static case is answering queries with respect to a \emph{sliding window} over a data stream. The sliding window model was extensively studied for multiple problems including summing~\cite{DBLP:conf/swat/Ben-BasatEFK16,DBLP:journals/siamcomp/DatarGIM02}, heavy hitters~\cite{DBLP:journals/percom/Ben-BasatEF18,DBLP:conf/infocom/Ben-BasatEFK16}, Bloom filters~\cite{DBLP:conf/infocom/AsafBEF18} and counting distinct elements~\cite{fusy2007estimating}.
\fi

{\bf Algorithms that Sum over Sliding Windows. }
Our \suffixsum{} queries for streaming are a generalization of the problem of summing over sliding windows. That is, window summation is a special case of the suffix sum problem where the algorithm is always asked for the sum of the last $i \le n$ elements.
Approximating the sum of the last $n$ elements over a stream of integers in $\{0,1,\ldots,\ell\}$, was first introduced by Datar et al.~\cite{DBLP:journals/siamcomp/DatarGIM02}. 
They proposed a $(1+\eps)$ multiplicative approximation algorithm that uses $O\parentheses{\oneOverE\parentheses{\lg^2 n +\lg\bsrange\cdotpa{\lgw+\lg\lg\bsrange}}}$ bits and operates in amortized time $\Op{{\lg\bsrange}/{\lgw}}$ or $O(\lg (\bsrange\cdot n))$ worst case.
In~\cite{GibbonsT02}, Gibbons and Tirthapura presented a $(1+\eps)$ multiplicative approximation algorithm that operates in constant worst case time while using similar space for $\bsrange=\window^{O(1)}$.
\cite{DBLP:conf/swat/Ben-BasatEFK16} studied the potential memory savings one can get by replacing the $(1+\eps)$ multiplicative guarantee with a $\lnrErrSymbol$ additive approximation. They showed that $\Thetap{{\bsrange\cdot n}/{\lnrErrSymbol}+\lgw}$ bits are required and sufficient. 
Recently, \cite{DBLP:conf/mfcs/Ben-BasatEF18} showed the potential memory saving of a bi-criteria approximation, which allows error in both the sum and the time axis, for sliding window summation. 
\cite{DBLP:journals/corr/abs-1804-10740} looks at a generalization of the ssA queries to general alphabet, where at query time we also receive an element 
$x$ and return an estimate for the frequency of $x$ in the last $i$ elements.

It is worth mentioning that these data structures \emph{do} allow computing the sum of a window whose size is given at the query time. Alas, the query time will be slower as they do not keep aggregates that allow quick computation. 
Specifically, we can compute a $(1+\epsilon)$ multiplicative approximation in $O(\oneOverE\lg (\ell n\eps))$ time using the data structures of \cite{DBLP:journals/siamcomp/DatarGIM02} and \cite{GibbonsT02}. We can also use the data structure of~\cite{DBLP:conf/swat/Ben-BasatEFK16} for an additive approximation of $\lnrErrSymbol$ in $O(n\ell/\delta)$ time.

\begin{table}
\centering
\scalebox{0.9}{
\begin{tabular}{c | c | c | c}
\hline
Query & Space (in bits) & Query time & Error\\
\hline
\multicolumn{4}{c}{Lower bounds}\\
\hline
$\rankA{}_1$ , $\vselectA{}_1$& $\floor{n/\delta}$ &    &\multirow{4}{*}{$\delta$, additive}\\\cline{1-3}
$\rankA{}_1$, $\vselectA{}_1$ &  $\mathcal{B}(\floor{n/\delta},\floor{m/\delta})$ &    &\\\cline{1-3}
$\vrankA{}_1$, $\dselectA{}_1$ & $\floor{n/2\delta}\lg{\delta}$ &   & \\
$\dselectA{}_1$ &  $O((n/\delta)\lg^{O(1)}{\delta})$ & $\Omega(\lg{\lg{n}})$ & \\
\hline
\multicolumn{4}{c}{Upper bounds}\\
\hline
$\rankA{}_1$, $\vselectA_1$ &  $n/\delta+o(n/\delta)$ & \multirow{3}{*}{$O(1)$}   &\multirow{4}{*}{$\delta$, additive}\\\cline{1-2}
$\rankA{}_1$ , $\vselectA_1$&   $\mathcal{B}(n/\delta, m/\delta)+o(n/\delta)$ &  &\\\cline{1-2}
$\vrankA{}_1$ &  $(n/\delta)\lg{\delta}+o((n/\delta)\lg{\delta})$ & & \\\cline{1-3}
$\dselectA{}_1$ &  $(n/\delta)\lg{\delta}+o((n/\delta)\lg{\delta})$ & $t(n/\delta, n)$ & \\
\hline
\end{tabular}}
\caption{Summary of results of upper and lower bounds for approximate $\rank{}$ and $\select{}$ queries on bit-string of size $n$ ($m$ is the number of $1$'s in $B$). The function $t(n,u)$ is defined as $t(n,u) = O(\min\{\lg\lg n \lg{\lg{u}}/\lg{\lg{\lg{u}}}, \sqrt{\lg{n}/\lg{\lg{n}}}\})$.
}
\vspace*{-0.5cm}
\label{tab:summary}
\end{table}

\begin{table*}[]
	\centering
	\scriptsize
	\hspace*{-0.2cm}	
\scalebox{0.9}{	
	\begin{tabular}{|c|c|c|c|c|}
		\hline
		& Guarantee & Space (in bits) & Update Time & Query time \\\hline
		DGIM02~\cite{DBLP:journals/siamcomp/DatarGIM02}  & $(1+\eps)$-multiplicative & $O(\oneOverE\lg(\ell n)\lg (n\lg \ell))$ & $O(\lgp{\ell n})$ & $O(\oneOverE\lgp{\ell n\eps})$\\\hline
		GT02~\cite{GibbonsT02}  & $(1+\eps)$-multiplicative & $O(\oneOverE\lg^2(\ell n))$ & $O(1)$ & $O(\oneOverE\lgp{\ell n\eps})$\\\hline											
		BEFK16~\cite{DBLP:conf/swat/Ben-BasatEFK16}  & $\lnrErrSymbol$-additive, for $\lnrErrSymbol=\Omegap{\ell}$ & $\Thetap{{\bsrange\cdot n}/{\lnrErrSymbol}+\lgw}$ & $O(1)$ & $O({\bsrange\cdot n}/{\lnrErrSymbol})$\\\hline		
		BEFK16~\cite{DBLP:conf/swat/Ben-BasatEFK16}  & $\lnrErrSymbol$-additive, for $\lnrErrSymbol=o\parentheses{\ell}$ & $\Thetap{n\lgp{{\ell}/{\lnrErrSymbol}}}$ & $O(1)$ & $O(n)$\\\hline				
		This paper  & $\lnrErrSymbol$-additive & Same as in~\cite{DBLP:conf/swat/Ben-BasatEFK16} & $O(1)$ & $O(1)$\\\hline						
		
	\end{tabular}}\smallskip
	\normalsize
	\caption{Comparison of data structures for \suffixsum{} queries over stream of integers in $\{0,\dots ,\ell\}$. All works can answer fixed-size window queries (where $i\equiv n$) in $O(1)$ time. Worst case times~are~specified.}
	\label{tbl:comp}
	\vspace{-0.5cm}
\end{table*}

\subsection{Our results}

In this paper, we obtain the following results for the approximate $\rank{}$, $\select{}$, $\suffixsum{}$ and $\iss{}$ queries with additive error. 
Let $B$ be a bit-string of size $n$.

\noindent{\bf 1. $\rank{}$ and $\select{}$ queries with additive error \bm{$\delta$}:}
In this case, we first show that $\floor{n/\delta}$ bits are necessary for answering $\rankA_1{}$ and $\vselectA_1{}$ queries on $B$ and propose a $(\ceil{n/\delta}+o(n/\delta))$-bit data structure that supports $\rankA_1{}$ and $\vselectA_1{}$ queries on $B$ in constant time. 
For the case when there are $m$ $1$'s in $B$, we show that
$\mathcal{B}(\floor{n/\delta},\floor{m/\delta})$ bits are necessary for answering $\rankA_1{}$ and $\vselectA_1{}$ queries on $B$, and obtain $\mathcal{B}(\floor{n/\delta}, \floor{m/\delta})+o(n/\delta)$-bit data structure that supports $\rankA_1{}$ and $\vselectA_1{}$ queries on $B$ in constant time.
For $\vrankA{}_1$ and $\dselectA{}_1$ queries on $B$, we show that $\floor{n/2\delta}\lg{\delta}$ bits are  necessary for answering both queries, and obtain an $(n/\delta)\lg{\delta}+o((n/\delta)\lg{\delta})$-bit data structure that supports $\vrankA{}_1$ queries in $O(1)$ time, and $\dselectA_1{}$ queries in 
$O(\min\{\lg{\lg{(n/\delta)}}\lg{\lg{n}}/\lg{\lg{\lg{n}}}, \sqrt{\lg{(n/\delta)}/\lg{\lg{(n/\delta)}}}\})$ time.
Furthermore, we show that there exists an additive error $\delta$ such that
any $O((n/\delta)\lg^{O(1)}{\delta})$-bit data structure requires at least $\Omega(\lg{\lg{n}})$ time to answer $\dselectA{}_1$ queries on $B$.

Using the above data structures, 
we also obtain data structures for answering approximate $\rank{}$ and $\select{}$ queries on a given multiset $S$ from the universe $U = \{1,2 \dots n\}$ with additive error $\delta$, where $\rank{}(i, S)$ query returns the value $| \{ j \in S | j \le i \} |$,  and $\select{}(i, S)$ query returns the $i$-th smallest element in $S$.
We consider two different cases: (i) \vrankA{}, \rankA{} \vselectA{}, and \dselectA{} queries when $|S| = m$, and (ii) \rankA{} and \vselectA{} queries when the frequency each elements in $S$ is at most $\ell$. 
Furthermore for case (ii), we first show that at least $\floor{n/\ceil{\delta/\ell}}\lg{(\max{(\floor{\ell/\delta},1)}+1)}$ bits are necessary for answering $\rankA{}$ queries, and obtain an optimal space structure that supports $\rankA{}$ queries in constant time, and
an asymptotically optimal space structure that supports both $\rankA$ and $\vselectA$ queries in constant time when $\ell = O(\delta)$. 

We also consider the $\rankA{}$ and $\vselectA{}$ queries on strings over large alphabets.
Given a string $A$ of length $n$ over the alphabet $\Sigma=\{1,2, \dots, \sigma\}$ of size $\sigma$, we obtain 
an $(2n/\delta\lg{(\sigma+1)}+o((n/\delta)\lg{(\sigma+1)})$-bit data structure that supports \rankA{} and \vselectA{} on $A$ in $O(\lg{\lg{\sigma}})$~time.
We summarize our results for bit-strings in Table~\ref{tab:summary}.
\\\\
\noindent{\bf 2. \suffixsum{} and \iss{} queries with additive error \bm{$\delta$}:}
We first consider a data structure for answering $\suffixsum{}$ and $\iss{}$ queries on binary stream, i.e., 
all integers in the stream are $0$ or $1$. 
For exact $\suffixsum{}$ and $\iss{}$ queries on the stream, 
we  propose an $n+o(n)$-bit data structure for answering those queries in constant time 
while supporting constant time updates whenever a new element arrives from the stream.
This data structure is obtained by modifying
 the data structure of Clark and Munro~\cite{Clark:1996:EST:313852.314087} for answering \rank{} and \select{} queries on bit-strings.
Using the above structure, we obtain an $(n/\delta+o(n/\delta)+O(\lg{n}))$-bit structure that supports $\ssA{}$  and $\issA{}$ 
queries on the stream in constant time while supporting constant time updates. 
Since at least $\floor{n/\delta}$ bits are necessary for answering $\rankA_1{}$ (or $\vselectA_1{}$) queries on bit-strings, 
and $\floor{\lg{n}}$ bits are necessary for answering $\suffixsum (n,n)$ queries~\cite{DBLP:conf/swat/Ben-BasatEFK16}, 
the space usage of our data structure is succinct (i.e., optimal upto lower-order terms) when $n/\delta = \omega(\lg{n})$, and asymptotically optimal otherwise.

We then consider the generalization that allows integers in the range $\{0,1,\ldots,\ell\}$, for some $\ell\in\mathbb N$. 
First, we present an algorithm that uses the optimal $n\lgp{\ell+1}(1+o(1))$ bits for exact suffix sums. Then, we provide a second algorithm that uses $\floor{n/\ceil{\delta/\ell}}\lg{(\max{(\floor{\ell/\delta},1)}+1)}(1+o(1)) + O(\lg n)$ bits for solving \ssA{}. Specifically, our data structure is succinct when $n/\delta=\omega(\lg n/\ell)$, and is asymptotically optimal otherwise, and improves the query time of \cite{DBLP:conf/swat/Ben-BasatEFK16} while using the same space.
Table~\ref{tbl:comp} presents this comparison.

\section{Queries on bit-strings}
\label{sec:additive}
In this section, we first consider the data structures for answering approximate \rank{} and \select{} queries on bit-strings and multisets. 
We also show how to extend our data structures on static bit-strings to the sliding windows on binary streams, for answering 
approximate \suffixsum{} and \iss{} queries.

\subsection{Approximate \rank{} and \select{} queries on bit-strings} 
We now consider the approximate \rank{} and \select{} queries on bit-strings with additive error $\delta$. 
We only show how to support $\vrankA_1$, $\rankA_1$, $\dselectA_1$, and $\vselectA_1$ queries.
To support $\vrankA_0$, $\rankA_0$, $\dselectA_0$, and $\vselectA_0$ queries, one can construct the same data structures
on the bit-wise complement of the original bit-string. 
We first introduce a few previous results which will be used in our structures.
The following lemmas describe the optimal structures for supporting \rank{} and \select{} queries on bit-strings.

\begin{lemma}[\cite{Clark:1996:EST:313852.314087}]
\label{lem:clark}	
For a bit-string $B$ of length $n$, there is a data structure of size $n + o(n)$ bits that supports $\rank_0{}$, $\rank_1{}$, $\select_0{}$, and $\select_1{}$ queries in $O(1)$ time.
\end{lemma}
   
\begin{lemma}[\cite{DBLP:journals/talg/RamanRS07}]
\label{lem:RRR}	  
For bit-string $B$ of length $n$ with $m$ 1's, there is a data structure of size 
\begin{itemize}
\item{(a)} $\mathcal{B}(n,m) + o(m)$ bits that supports $\select_1{}$ query in $O(1)$ time, and
\item{(b)} $\mathcal{B}(n,m) + o(n)$ bits that supports $\rank_0{}$, $\rank_1{}$, $\select_0{}$, and $\select_1{}$ queries in $O(1)$~time.
\end{itemize}
\end{lemma}

We use results from~\cite{DBLP:journals/tcs/HonSS11} and~\cite{DBLP:conf/wads/RamanRR01}, which describe efficient data structures for supporting the following queries on integer arrays. For a standard word-RAM model with word size $O(\lg{U})$ bits, let $A$ be an array of $n$ non-negative integers. For $1 \le i \le n$ and any non-negative integer $x$, (i) $\summ{}(i)$ returns the value $\sum_{j = 1}^{i} A[j]$, and (ii) $\search{}(x)$ returns the smallest $i$ such that $\summ{}(i) > x$.
We use the following function to state the running time of some of the ({\em Searchable Partial Sum}) queries in the lemma below, and in the rest of the paper.

$$SPS(n,U)  = \left\{\begin{array}{ll}
\textrm{$O(1)$} & \textrm{\hspace{0.3cm}if $n=polylog(U)$}\\
O(\min{ \{ \lg\lg n \lg\lg U / \lg\lg\lg U, \sqrt{ \lg n/ \lg\lg n } \} } ) & \textrm{\hspace{0.3cm}otherwise}\\
\end{array} \right.$$

\begin{lemma}[\cite{DBLP:journals/tcs/HonSS11}, \cite{DBLP:conf/wads/RamanRR01}]
\label{lem:prefix}
An array of $n$ non-negative integers, each of length at most $\alpha$ bits, can be stored using $\alpha n+o(\alpha n)$ bits, 
to support $\summ{}$ queries on $A$ in constant time, and $\search{}$ queries on $A$ in $SPS(n,n2^{\alpha})$ time. 
Moreover, when $\alpha = O(\lg\lg{n})$, we can answer both queries in $O(1)$ time. 
\end{lemma}

\noindent\textbf{Supporting \rankA{} and \vselectA{} queries. }
We first consider the problem of supporting $\rankA_1{}$ or $\vselectA_1$ 
queries with additive error $\delta$ on a bit-string $B$ of length $n$. 
We first prove a lower bound on space used by any data structure
that supports either of these two queries.

\begin{theorem}
\label{thm:lb1}
Any data structure that supports $\rankA{}_1$ or $\vselectA_1$ queries with 
additive error $\delta$ on a bit-string of length $n$ requires at least 
$\floor{n/\delta}$ bits. Also if the bit-string has $m$ 1's in it, then
at least $\mathcal{B}(\floor{n/\delta},\floor{m/\delta})$ bits are necessary for 
answering the above queries.
\end{theorem}
\begin{proof}
Consider a bit-string $B$ of length $n$ divided into $\floor{n/\delta}$ blocks $B_1$, $B_2$, \dots $B_{\floor{n/\delta}}$ such that
for $1 \le i < \floor{n/\delta}$, $B_i = B[\delta(i-1)+1 \dots \delta i]$ and 
$B_{\floor{n/\delta}} = B[\delta(\floor{n/\delta}-1)+1 \dots n]$ (the last block may contain more than $\delta$ bits). 
Let $S$ be the set of all possible bit-strings satisfying the condition that all the bits within a block are the same (i.e., either all zeros or all ones).
Then it is easy to see that $|S| = 2^{\floor{n/\delta}}$.
We now show that any two distinct bit-strings in $S$ will have different 
answers for some $\rankA_1{}$ query (and also some $\vselectA_1{}$ query).
Consider two distinct bit-strings $B$ and $B'$ in $S$, and let $i$ be the index 
of the leftmost block such that $B_{i} \neq B'_{i}$. 
Then it is easy to show that there is no value which is the answer of both $\rankA_1{}(i \delta, B, \delta) $ and $\rankA_1{}(i \delta, B', \delta)$ queries
and also there is no position of $B$ which is the answer of both $\vselectA_1{}(j, B, \delta)$ and $\vselectA_1{}(j, B', \delta)$ queries, where $j$ is the number of 1's in $B[1 \dots i \delta]$. 
Thus any structure that supports either of these queries must distinguish between 
every element in $S$, and hence $\floor{n/\delta}$ bits are necessary to answer $\rankA_1$ or $\vselectA_1$ queries.

For the case when the number of $1$'s in the bit-string is fixed to be $m$, 
we choose $\floor{m/\delta}$ blocks from each bit-string and make all bits in 
the chosen blocks to be $1$'s (and the rest of the bits as $0$'s). 
Since there are $\floor{n/\delta} \choose \floor{m/\delta}$ ways for select 
such $\floor{m/\delta}$ blocks in a bit-string of length $n$, it implies that $\mathcal{B}(\floor{n/\delta},\floor{m/\delta})$ 
bits are necessary to answer $\rankA_1$ and $\vselectA_1$ queries in this case.
\end{proof}

Now we describe a data structure for supporting $\rankA_1$ and $\vselectA_1$ queries in constant time, using optimal space.

\begin{theorem}\label{thm:ub1}
For a bit-string $B$ of length $n$, there is a data structure that uses $n/\delta+o(n/\delta)$ bits and supports $\rankA{}_1$ and $\vselectA_1$ queries with additive error $\delta$, in constant time. If there are $m$ 1's in $B$, the data structure uses $\mathcal{B}(n/\delta, m/\delta)+o(n/\delta)$ bits and supports the queries in $O(1)$ time.
\end{theorem}
\begin{proof}
We divide the $B$ into $\ceil{n/\delta}$ blocks $B_1$, $B_2$, \dots $B_{\ceil{n/\delta}}$ such that
for $1 \le i < \ceil{n/\delta}$, $B_i = B[\delta(i-1)+1 \dots \delta i]$ and 
$B_{\ceil{n/\delta}} = B[\delta(\ceil{n/\delta}-1)+1 \dots n]$. 
Now we define a new bit-string $B'$ of length $\ceil{n/\delta}$ such that for $1 \le i \le \ceil{n/\delta}$, 
$B'[i] = 1$ if $B_i$ contains $j\delta$-th 1 in $B$ for any $j \le i$, and otherwise $B'[i]=0$ 
(note that for any $1 \le j \le \ceil{n/\delta}$, any block of $B$ has at most one position of $j\delta$-th $1$ in $B$).
By Lemma~\ref{lem:clark}, we can support $\rank_1$ and $\select_1$ queries on $B'$ in constant time, using $n/\delta+o(n/\delta)$ bits. 
Now we claim that $C = \delta \cdot \rank{}_1(\floor{i/\delta})+(i \mod \delta)B'[\ceil{i/\delta}]$ gives an answer of the $\rankA_1{}(i, B, \delta)$ query. 
Let $D = \delta \cdot \rank{}_1(\floor{i/\delta})$, and let $d$ be the position of $D$-th 1 in $B$. 
From the definition of $B'$, we can easily show that if $B'[\ceil{i/\delta}] = 0$ or $(i \mod \delta) = 0$, the claim holds 
since there are less than $\delta$ 1's in $B[d \dots i]$.
Now consider the case when $B'[\ceil{i/\delta}] = 1$ and $(i\mod \delta) \neq 0$. 
Then there are at most $(\delta+(i \mod \delta)-1)$ 1's in $B[d \dots i]$ 
when $(\delta \floor{i/\delta}+1)$ is the position of the $(D+\delta)$-th 1 in $B$, and
all the values in $B[(\delta \floor{i/\delta}+2) \dots i]$ are 1. 
Also there are at least $\delta-(\delta-(i \mod \delta)) = (i \mod \delta)$ 1's in $B[d \dots i]$ when $(\delta \ceil{i/\delta})$ is the position of the $(D+\delta)$-th 1 in $B$ and all the values in $B[\delta \floor{i/\delta}+ (i \mod \delta)+1 \dots \delta \ceil{i/\delta}]$ are 1. 
By the similar argument, we can show that one can answer the $\vselectA_1(i, B, \delta)$ query in $O(1)$ time by returning $\delta(\select{}_1(\floor{i/\delta}, B')-1)+(i \mod d)$.

Finally, in the case when there are $m$ 1's in $B$, there are at most $\floor{m/\delta}$ 1's in $B'$. 
Therefore by Lemma~\ref{lem:RRR}(b), we can support $\rankA_1$ and $\vselectA_1$ queries (as before) 
in $O(1)$ time, using $\mathcal{B}(\ceil{n/\delta}, \floor{m/\delta})+o(n/\delta)$ bits.
\end{proof}

Note that in the above proof, we can answer $\rankA_1$ (or $\vselectA_1$) queries on $B$ 
using any data structure that supports $\rank{}_1$ (or $\select{}_1$) queries on $B'$. 
Thus, if $B$ is very sparse, i.e., when $\mathcal{B}(n/\delta, m/\delta) \ll o(n/\delta)$ 
(in this case, the space usage of the structure of Theorem~\ref{thm:ub1} is sub-optimal), 
one can use the structure of~\cite{DBLP:conf/alenex/OkanoharaS07}
that uses $(m/\delta)\lg(n/m)+O(m/\delta)$ bits (asymptotically optimal space), 
to support $\rankA_1$ queries in $O(\min\{\lg{m},\lg{(n/m)}\})$ time, and $\vselectA_1$ queries in constant time.
\\\\
\noindent\textbf{Supporting $\vrankA{}$ and $\dselectA{}$ queries. }
Now we consider the problem of supporting $\vrankA_1$ and $\dselectA_1$ queries with additive error $\delta$ on bit-strings of length $n$. 
The following theorem describes a lower bound on space. 

\begin{theorem}
(*)
\ifdefined\shortversion\footnote{Proofs of the results marked (*) are omitted due to space limitation, and will appear in the full version.}
\else
\footnote{Proofs of the results marked (*) are deferred to the appendix.}
\fi
\label{thm:ldselect}
Any data structures that supports $\vrankA_1$ or $\dselectA_1$ queries with additive error $\delta$ on a bit-string of length $n$ requires at least $\floor{n/2\delta}\lg{\delta}$ bits. 
\end{theorem}

We now show that for some values of $\delta$, any data structure that uses up to a $\lg^{O(1)} \delta$ 
factor more than the optimal space cannot support $\dselectA_1$ queries in constant time. 

\begin{theorem}\label{thm:dselectLgLgLowerBound}(*)
Any $((n/\delta) \lg^{O(1)}{\delta})$-bit data structure that supports $\dselectA_1$ 
queries with an additive error $\delta = O(n^c)$, for some constant $0 < c \le 1$ on a 
bit-string of length $n$ requires $\Omega(\lg{\lg{n}})$ query time.
\end{theorem}

The following theorem describes a simple data structure for supporting $\dselectA_1$ queries. 

\begin{theorem}(*)
\label{thm:ubds}
For a bit-string $B$ of length $n$, there is a data structure of size $(n/\delta)\lg{\delta}+o((n/\delta)\lg{\delta})$ bits, which supports $\vrankA_1$ queries on $B$ using $O(1)$ time and $\dselectA_1{}$ queries on $B$ using $SPS(n/\delta, n)$ time. 
\end{theorem}

\subsection{Approximate \rank{} and \select{} queries on multisets}
In this section, we describe data structures for answering approximate \rank{} and \select{} queries on a multiset with additive error $\delta$. 
Given a multiset $S$ where each element  is from the universe $U = \{1,2 \dots n\}$, the  $\rank$ and $\select$ queries on $S$ are defined as follows.

\begin{itemize}
\item{$\rank{}(i, S)$:} returns the value $| \{ j \in S | j \le i \} |$.
\item{$\select{}(i, S)$:} returns the $i$-th smallest element in $S$.
\end{itemize}

One can define approximate \rank{} and \select{} queries on multisets
(also denoted as \vrankA{}, \rankA{}, \vselectA{}, \dselectA{}) analogously to the queries on strings.
Any multiset $S$ of size $m$ can be represented as a \textit{characteristic vector} 
$B_S$ of size $m+n$, such that
$B_S = 1^{m_1}01^{m_2}0\dots1^{m_n}0$ when the element $k$ has multiplicity $m_k$ in the multiset $S$, for $1 \le k \le n$.
It is easy to show that by answering $\rank{}_{b}$ and $\select{}_{b}$ queries on $B_S$, for $b \in \{0,1\}$, 
one can answer $\rank{}$ and $\select{}$ queries on $S$. 
We now describe efficient structures for the following two cases.
\\\\
\noindent
{\bf (1) \vrankA{}, \rankA{}, \vselectA{}, and \dselectA{} queries when \bm{$|S| = m$} is fixed:} 
We construct a new string $B'_S$ of length $\floor{m/\delta}+n$ such that $B'_S$ only keeps
every $i\delta$-th $1$ from $B_S$, for $1 \le i \le \floor{n/\delta}$ (and removes all other $1$'s). 
To answer the query $\rankA(i,S,\delta)$, we first compute $\select{}_0(i, B'_S)-i = \floor{\rank(i, S)/\delta}$,
and return $\delta(\select{}_0(i, B'_S)-i)$ as the answer.
It is easy to see that $\delta\floor{\rank(i, S)/\delta}$ is an answer to the $\rankA(i, S, \delta)$ query. 
Similarly, we can answer the $\vselectA(i, S,\delta)$ query by returning $\rank_{0}(\select{}_1(\floor{i/\delta}, B'_S), B'_S)+1$.
We represent $B'_{S}$ using the structure of Lemma~\ref{lem:RRR}(b), which uses
$\mathcal{B}(n+\floor{m/\delta}, \floor{m/\delta})+o(n+\floor{m/\delta})$ bits and 
supports $\rank_0$, $\rank_1$, $\select_0$ and $\select_1$ queries on $B'_S$ in constant time.
Thus, both $\rankA$ and $\vselectA$ queries on $S$ can be supported in constant time.

For answering $\vrankA$ and $\dselectA$ queries on $S$, we first construct the data structure 
of Theorem~\ref{thm:ubds} to support $\dselectA_1$ queries on $B_S$. 
In addition, 
we maintain the data structure of Lemma~\ref{lem:prefix} to support $\summ$ and $\search$ queries 
on arrays $D[1 \dots \ceil{(n+m)/\delta}]$ and $E[1 \dots \ceil{(n+m)/\delta}]$ which are defined as follows.
For $1 \le i \le \ceil{(n+m)/\delta}$, $D[i]$ and $E[i]$ stores the number of $0$'s and $1$'s in the block $B_{S_i}$ respectively
(as defined in the proof of Theorem~\ref{thm:ubds}). 
By Lemma~\ref{lem:prefix} and Theorem~\ref{thm:ubds}, the total space for this data structure is $O((n'/\delta)\lg{\delta})$ 
bits. 
To answer $\vrankA(i, S, \delta)$, we first find the block $B_{S_{j}}$  of $B_S$ which contains $i$-th $0$ by answering $\search(i)$ query on $D$, and then return $\summ(j-1)$ query on $E$. 
To answer $\dselectA(i, S, \delta)$, we first find the block $B_{S_{j}}$ of $B_S$ which contains an 
answer of the $\dselectA_1(i, B_S, \delta)$ query, and then return $\summ(j-1)$ as the answer for $\dselectA(i, S, \delta)$. Note that if $j=1$, we return $0$ for both queries.
The total running time is $SPS(n'/\delta, n')$ for both $\vrankA$ and $\dselectA$ queries on $S$,
by Lemma~\ref{lem:prefix} and Theorem~\ref{thm:ubds}. For special case 
when $\min\{(n+m)/\delta, \delta\} = polylog(n+m)$, we can answer $\vrankA$ and $\dselectA$ queries on $S$ in constant time.
\\\\
\noindent
{\bf (2) \rankA{} and \vselectA{} queries when the frequency of each element in $S$ is at most $\ell$:}
We first show that at least $\floor{n/\ceil{\delta/\ell}}\lg({\max{(\floor{\ell/\delta},1)}+1})$ bits are are necessary for supporting $\rankA$ queries on $S$. 

\begin{theorem}\label{thm:largeAlphabetLB}(*)
Given a multiset $S$ where each element is from the universe $U =\{1, 2, \dots , n\}$ of size $n$, any data structure that supports $\rankA$ queries on $S$ requires at least $\floor{n/\ceil{\delta/\ell}}\lg{(\max{(\floor{\ell/\delta},1)}+1)}$ bits, where $\ell$ is a bound on the  maximum frequency of each element in $S$.
\end{theorem}

We describe a data structure which answers $\rankA{}$ and $\vselectA{}$ queries on $S$ in $O(1)$ time. For $\rankA{}$ queries, it uses the optimal space. 
The details are
\ifdefined\shortversion omitted due to space limitation.
\else
described in Appendix~\ref{sec:multisetupper}.
\fi
\subsection{Approximate \suffixsum{} and \iss{} queries on binary streams}
\label{sec:binarysuffixsum}
In this section, we consider a data structure for answering $\dsuffixA$ and $\vissA$ queries
on a binary stream. 
We first show how to modify the data structure of the Lemma~\ref{lem:clark}, for answering $\suffixsum (i, n)$ and $\iss (i, n)$ queries in constant time using $n+o(n)$ bits, while supporting updates in constant time. 
We break the stream into \textit{frames}, 
which is $n$-bit consecutive elements in the stream.
Since one can construct a data structure of Lemma~\ref{lem:clark} in online~\cite{Clark:1996:EST:313852.314087}, 
it is easy to show that we can answer $\suffixsum{}$ and $\iss{}$ queries in constant time using $2n+o(n)$ bits while supporting constant-time updates 
by maintaining two data structure of Lemma~\ref{lem:clark} such as
one for the current frame and other for the previous frame of the stream. 
To make this data structure using $n+o(n)$ bits, 
we construct a data structure of Lemma~\ref{lem:clark} on the new frame  
while replacing the oldest part of the data structure constructed on the previous frame.
The details of the succinct data structure
\ifdefined\shortversion are omitted due to space limitation.
\else
are described in Appendix~\ref{subsec:nonapprox}.
\fi

Next, we consider a data structure for answering  $\dsuffixA{} (i, n, \delta)$ and $\vissA{} (i, n, \delta)$ queries on the binary stream in constant time using $\ceil{n/\delta}+O(\lg{n})+ o(n/\delta)$ bits.
We first split each frame $f = f_1 \dots f_n$ into $\ceil{n/\delta}$ chunks 
$g_1 \dots g_{\ceil{n/\delta}}$ such that for $1 \le i \le \ceil{n/\delta}$,
$g_i = 1$ if and only if $f_{(i-1)\delta+1} \dots f_{\min{(n,i\delta)}}$ 
contains $j\delta$-th 1 in $f$ for any integer $j \le i$.
Now consider a (virtual) binary stream of $g_i$'s. 
Then we can construct an $\ceil{n/\delta}+o(n/\delta)$-bit data structure for answering $\suffixsum{} (i, n)$, $\iss{} (i, n)$ queries in constant time while supporting constant-time updates on the such stream
(In the rest of this section, all of $\suffixsum{}$ and $\iss{}$ queries are answered on the virtual stream). 
We also maintain $c$ and $tc$, 
which stores the number of 1's in the current frame and chunk of the stream respectively.
Finally, we maintain an value $t$ which is an index of the last-arrived element in the current frame. All these additional values can be stored using $O(\lg{n})$ bits.

When $f_t$ is arrived,
We first increase $c$ and $tc$ by 1 if $f_t = 1$.  
If $(t \mod \delta) = 0$ or $t = n$, we send $1$ to the virtual stream if
there is an integer $j \le t$ such that $c-tc \le j\delta \le c$, and send $0$ to the virtual stream otherwise.
After that, we update the data structure which supports $\suffixsum{}$ and $\iss{}$ queries on the virtual stream, and reset $tc$ to zero (if $t= n$, we also reset $c$ to zero).
Since we can update the data structure on the virtual stream in constant time,
the above procedure can be done in constant time.
Now we describe how to answer $\dsuffixA{}$ and $\vissA{}$ queries.
\begin{itemize}
\item{$\dsuffixA{}$ queries:}
To answer the $\dsuffixA{}(i,S, \delta)$ query, we return $0$ if  $i \le \delta$.
If not, let $f'_{i}$ be the $(\ceil{(i-(t \mod \delta))/\delta})$-th last element in the virtual stream, Then we return $tc + \delta\suffixsum (\floor{(i-(t \mod \delta))/\delta}, \ceil{n/\delta})+(i-(t \mod \delta) \mod \delta)f'_{i}$,  
which gives an answer of the $\dsuffixA{} (i,n, \delta)$ query by the same argument as the proof of Theorem~\ref{thm:ub1}.
\item{$\vissA{}$ queries:}
To answer the $\issA{}(i,n, \delta)$ query, we return $n-(t-(t \mod \delta))$ if  $i \le tc$.
Otherwise, we return $n-(\delta(\iss{} (\floor{(i-tc)/\delta}, \ceil{n/\delta})+((i-tc) \mod \delta))$
by the same argument as the proof of Theorem~\ref{thm:ub1}.
\end{itemize}
Since $\suffixsum{}$ and $\iss{}$ queries on the virtual stream take $O(1)$ time, 
we can answer both $\dsuffixA{}$ and $\vissA{}$ queries on the stream in $O(1)$ time.
Thus we obtain the following theorem.
\begin{theorem}
\label{thm:binarystream}
For a binary stream, there exists a data structure that uses $\ceil{n/\delta}+O(\lg{n})+o(n/\delta)$ bits and supports $\dsuffixA{}$ and $\vissA{}$ queries on the stream with additive error $\delta$, in constant time. Also, the structure supports updates in constant time.
\end{theorem}

Comparing to the lower bound of Theorem~\ref{thm:lb1} for answering $\rankA{}$ and $\vselectA{}$ queries on bit-strings (this also gives a lower bound for answering $\dsuffixA{}$ and $\vissA{}$ queries), the above data structure takes $\Omega(n/\delta)$ bits when $n/\delta = o(\lg{n})$.
However in the sliding window of size $n$, at least
$\floor{\lg{n}}$ bits are necessary~\cite{DBLP:conf/swat/Ben-BasatEFK16} for answering $\dsuffixA{}$  queries even the case when $i$ is fixed to $n$. 
Therefore the data structure of Theorem~\ref{thm:binarystream} supports $\dsuffixA$ and $\vissA$ queries with optimal space when $n/\delta = \omega(\lg{n})$, and asymptotically optimal otherwise.

\section{Queries on strings over large alphabet}
In this section, we consider non-binary inputs. First, we look at general alphabet and derive results for approximate \rank{} and \select{}. Then we consider suffix sums over integer streams.

\subsection{$\rankA{}$ and $\vselectA{}$ queries on strings over general alphabet}
Let $A$ be a string of length $n$ over the alphabet $\Sigma=\{1,2, \dots, \sigma\}$ of size $\sigma$.
Then, for $1 \le j \le \sigma$, the query $\rank{}_{j}(i, A)$ returns the number of $j$'s in $A[1 \dots i]$, 
and the query $\select_j(i, A)$ returns the position of the $i$-th $j$ in $A$ (if it exists).
Similarly, the queries $\rankA{}_{j}(i, A, \delta)$ and $\vselectA{}_{j}(i, A, \delta)$ are
defined analogous to the queries $\rankA{}$ and $\vselectA{}$ on bit-strings.
One can easily show that at least $\floor{n/\delta}\lg{\sigma}$ bits are necessary 
to support $\rankA$ and $\vselectA$ queries on $A$, by extending the proof of 
Theorem~\ref{thm:lb1} to strings over larger alphabets.
In this section, we describe a data structure that supports $\rankA$ and $\vselectA$ queries 
on $A$ in $O(\lg{\lg{\sigma}})$ time, using twice the optimal space.
We make use of the following result from~\cite{Golynski:2006:ROL:1109557.1109599} for supporting $\rank$ and $\select$ queries on strings over large alphabets.
We now use the following lemma to prove our main result for the section.

\begin{lemma}[\cite{Golynski:2006:ROL:1109557.1109599}]
\label{lem:gmr}
Given a string of length $n$ over the alphabet $\Sigma=\{1,2, \dots, \sigma\}$,
one can support $\rank{}_j$ queries in $O(\lg{\lg{\sigma}})$ time and $\select_j$ 
queries in $O(1)$ time, using $n\lg{\sigma}+o(n\lg{\sigma})$ bits, for any $1 \le j \le \sigma$.
\end{lemma}

The following theorem shows we can construct a simple data structure for supporting 
$\rankA{}_j$ and $\vselectA{}_j$ queries on $A$ using the above lemma.

\begin{theorem}(*)
\label{thm:large}
Let $A$ be a string of length $n$ over the alphabet $\Sigma=\{1,2, \dots, \sigma\}$. Then for any $1 \le j \le \sigma$, 
one can support $\rankA{}_j$ and $\vselectA{}_j$ queries in $O(\lg{\lg{\sigma}})$ time
using  $2n/\delta\lg{(\sigma+1)}+o((n/\delta)\lg{(\sigma+1)})$ bits.
\end{theorem}

\subsection{Supporting \ssA{} queries over non-binary streams}
In this section, we consider the problem of computing suffix sums over a stream of integers in $\{1,2,\ldots,\ell\}$. 
This generalizes the result of the Theorem~\ref{thm:binarystream} for \ssA{}{}. 
For such streams, one can use \ssA{} binary search to solve \issA{}, while a constant time \issA{} queries are left as future work.
Specifically, we show a data structure that requires $\floor{n/\ceil{\delta/\ell}}\lg{(\max{(\floor{\ell/\delta},1)}+1)}(1+o(1))+O(\lg{n})$; i.e., it requires $1+o(1)$ times as many bits as the static-case lower bound of Theorem~\ref{thm:largeAlphabetLB} when $\delta = o(\ell \cdot n/\lg{n})$.

We note that this model was studied in~\cite{DBLP:conf/swat/Ben-BasatEFK16,DBLP:journals/siamcomp/DatarGIM02,GibbonsT02} for window-sum queries. That is, our work generalizes this model to allow the user to specify the window size $i\le n$ at query time while previous works only considered the sum of the last $n$ elements.
In fact, all previous data structure implicitly supports \ssA{} queries but with slower run time. 
\cite{GibbonsT02,DBLP:journals/siamcomp/DatarGIM02} requires $O(\epsilon^{-1}\lgp{\ell n\eps})$ time to compute a $(1+\eps)$ approximation for the sum of the last $n$ elements while~\cite{DBLP:conf/swat/Ben-BasatEFK16} needs $O\parentheses{{\ell\cdot n}/{\delta}}$ for a $\delta$-additive one. Here, we show how to compute a $\delta$-additive error for the sum of the last $i\le n$ elements in constant time for both updates and queries.
\\
\noindent\textbf{Exact \suffixsum{} queries}
En route to \ssA{}, we first discuss how to compute an exact answer for suffix sums queries. It is known, even for fixed window sizes, that one must use $n\lgp{\ell+1}$ bits for tracking the sum of a sliding window~\cite{DBLP:conf/swat/Ben-BasatEFK16}. Here, we show how to compute exact \ssA{} using succinct space of $n\lgp{\ell+1}(1+o(1))$ bits.

We start by discussing why the current approaches cannot work for a large $\ell$ value.
If we use sub-blocks of size $\Theta(\lg n)$ as in~\cite{Clark:1996:EST:313852.314087, Jacobson:1988:SSD:915547}, then the lookup table will require $(\ell+1)^{\Theta(\lg n)}=n^{\Theta(\lgp{\ell+1)}}$ bits, which is not even asymptotically optimal for non-constant $\ell$ values. While one may think that this is solvable by further breaking the sub-blocks into sub-sub-blocks, sub-sub-sub-blocks, etc., it is not the case. 
To see this, consider a lookup table for sequences of length $2$. Then its space requirement will be $(\ell+1)^2$ bits. If $\ell$ is large (say, $\ell\ge n$) then this becomes $\Omegap{n\ell}=\omega(n\lgp{\ell+1})$, which is not even asymptotically optimal.

\begin{theorem}(*)
There exists a data structure that requires $n\lgp{\ell+1}(1+o(1))$ bits and support constant time (exact) suffix sums queries and updates.
\end{theorem}
\renewcommand{\set}[1]{\left\{#1\right\}}
\noindent\textbf{General \ssA{} queries}
Here, we consider the general problem of computing \ssA{} (i.e., up to an additive error of $\delta$).
Intuitively, we apply the exact solution from the previous section on a compressed stream that we construct on the fly.
A simple approach would be to divide the streams into consecutive chunks of size $\max{(\floor\mu,1)}=\max{(\floor{\delta/\ell},1)}$ and represent each chunk's sum as an input to an exact suffix sum algorithm.
However, this fails to achieve succinct space. For example, summing $\ceil{\delta/\ell}$ integers requires $O(\ceil{\delta/\ell}\lgp{\ell+1})=\Omega(\lg \ell)$ bits. However, $\lg \ell$ bits may be asymptotically larger than the $\floor{n/\ceil{\mu}}\lg{(\max{(\floor{1/\mu},1)}+1)}$ bits lower bound of Theorem~\ref{thm:largeAlphabetLB}.

We alleviate this problem by \emph{rounding} the arriving elements. Namely, when adding an input $x\in\set{0,1,\ldots,\ell}$, we first round its value to $Round_{\nBits}(x)\triangleq2^{-\nBits}\ell\cdot\floor{\frac{x2^{\nBits}}{\ell}}$ so it will require $\nBits\triangleq\ceil{\lgp{n/\mu}+\lg\lg n}$ bits.
The rounding allows us to sum elements in a chunk (using a variable denoted by $\mathfrak{r}$), but introduces a rounding error. To compensate for the error, we both consider a smaller chunks; namely, we use chunks of size $\nu \triangleq \max\set{\floor{\mu\cdotpa{1-1/\lg n}},1}$. We also consider $\sensitivity\triangleq \floor{\lnrErrSymbol\cdotpa{1-1/\lg n}}$ that is slightly lower than $\lnrErrSymbol$ to compensate for the rounding error when $\mu\le 1$.~\footnote{If $\sensitivity=\minDelta$, then we simply apply the exact algorithm from the previous subsection.} 
We then employ the exact suffix sums construction from the previous section for window size of $s\triangleq \UBnBlocks$ (the number of chunks that can overlap with the window) over  a stream of integers in $\{1,\ldots,z\}$, where $z\triangleq\floor{\mu^{-1}\nu}$ is a bound on the~resulting~items. We use $\newInputLetter$ to denote the input that respresents the current block.

The query procedure is also a bit tricky. Intuitively, we can estimate the sum of the last $i$ items by querying \ranker{} for the sum of the last $i/\nu$ inserted values and multiplying the result by $\sensitivity$; but there are a few things to keep in mind. First, $i/\nu$ may not be an integer. Next, the values within the current chunk (that has not ended yet) are not recorded in \ranker{}. Finally, we are not allowed to overestimate, so $\mathfrak{r}$'s propagation may be an issue.

To address the first issue, we weigh the oldest chunk's $\newInputLetter$ value by the fraction of that chunk that is still in the window. For the second, we add the value of $\mathfrak{r}$ to the estimation, where $\mathfrak{r}$ is the sum of rounded elements. Notice that we do not reset the value of $\mathfrak{r}$ but rather propagate it between 
chunks.
Finally, to assure that our algorithm never overestimates we subtract $\sensitivity - 1/2$ from the result.
Our algorithm uses the following variables:
\begin{itemize}
	\item \ranker{} - an exact suffix sum algorithm, as described in the previous section. It allows computing suffix sums over the last $s= \UBnBlocks$ elements on a stream of integers in $\{1,\ldots,z\}$.	
	\item $\mathfrak{r}$ - tracks the sum of elements that is not yet recorded in \ranker{}.
	\item $o$ - the offset within the chunk.
\end{itemize}
A pseudo code of our method appears in Algorithm~\ref{alg:approxSliding}.
\begin{algorithm}[t]
	\caption{Algorithm for \ssA{}}\label{alg:approxSliding}
	\small
	\begin{algorithmic}[1]
		\State Initialization: $\mathfrak{r} \gets 0, o \gets 0, \ranker.\text{init()}$
		\Function{\add[\text{element }\inputVariable]}{}
		\State $\blockOffset \gets (\blockOffset+1)\mod \nu$
		\State $\mathfrak{r} \gets \mathfrak{r} + Round_{\nBits}(x)$				\label{line:sum}
		\If {$\blockOffset = 0$}\Comment{End of a chunk}\label{line:end-of-block}
		\State $\newInputLetter \gets \floor{\sensitivity^{-1} \cdot\mathfrak{r}}$ \label{line:setBit}
		\State $\mathfrak{r}\gets \mathfrak{r} - \sensitivity\cdot\newInputLetter$ \label{line:reduceS}
		\State $\ranker.\mbox{\sc Add}(\newInputLetter)$
		\EndIf
		\EndFunction		
		\Function{Query}{$i$}
		\If{$i \le o$} \Comment{Queried within the current chunk}
		\State \Return $\mathfrak{r} - \parentheses{\sensitivity - 1/2}$
		\Else
		\State $\numBits \gets \ceil{\frac{i-o}{\nu}}$
		\State $\setBits \gets \ranker.\mbox{\sc Query}\parentheses{\numBits}$
		\State $\lastBit \gets \setBits - \ranker.\mbox{\sc Query}\parentheses{\numBits-1}$
		\State $\outBits \gets \outBitsVal$
		\State \Return $\mathfrak r - \parentheses{\sensitivity - 1/2}+\sensitivity\cdot\setBits -\ell\cdot\lastBit\cdot\outBits$\label{line:est}
		\EndIf		
		\EndFunction		
	\end{algorithmic}
\end{algorithm}
Next follows a memory analysis of the algorithm. 
\begin{lemma}\label{lem:slidingQsrSpaceProof}(*)
	Algorithm~\ref{alg:approxSliding} requires $\sFactor\cdot\floor{n/\max{(\floor{\mu},1)}}\cdot\lg\big({\ceil{\mu^{-1}} + 1}\big) + O\parentheses{\lg n}$~bits.
\end{lemma}

Thus, we conclude that our algorithm is succinct if the error satisfies $\lnrErrSymbol=o\parentheses{{\ell\cdot n}/{\lg n}}$. We note that a $\floor{\lg n}$ bits lower bound for \bs{} with an additive error was shown in~\cite{DBLP:conf/swat/Ben-BasatEFK16}, even when only fixed sized windows (where $i \equiv n$) are considered. Thus, our algorithm always requires $O(\lnrLBSymbol)$ space, even if  $\lnrErrSymbol=\Omegap{{\ell\cdot n}/{\lg n}}$. Here, $\lnrLBSymbol=\floor{n/\ceil{\delta/\ell}}\lg{(\max{(\floor{\ell/\delta},1)}+1)}$ is the lower bound for static data shown in~Theorem~\ref{thm:largeAlphabetLB}.
\begin{corollary}
	Let $\ell,n,\lnrErrSymbol\in\mathbb N^+$ such that $\mu\triangleq \lnrErrSymbol/\ell$ satisfies $$\parentheses{\mu=o\parentheses{{n}/{\lg n}}} \wedge \brackets{(\mu=o(1))\vee (\mu=\omega(1)) \vee (\mu\in\mathbb N) \vee (\mu^{-1}\in\mathbb N)},$$ 
	then Algorithm~\ref{alg:approxSliding} is succinct. For other parameters, it uses  $O(\lnrLBSymbol)$ space.
\end{corollary}
We now state the correctness of our algorithm. 
\begin{theorem}\label{thm:slidingCorrecntess}(*)
	Algorithm~\ref{alg:approxSliding} solves \ssA{} while processing elements and answering queries in constant time.
\end{theorem}

\bibliography{ref}

\begin{thebibliography}{10}

\bibitem{DBLP:conf/infocom/AsafBEF18}
Eran Asaf, Ran Ben{-}Basat, Gil Einziger, and Roy Friedman.
\newblock Optimal elephant flow detection.
\newblock In {\em IEEE INFOCOM 2018}, pages 1--9, 2018.

\bibitem{DBLP:journals/percom/Ben-BasatEF18}
Ran Ben{-}Basat, Gil Einziger, and Roy Friedman.
\newblock Fast flow volume estimation.
\newblock {\em Pervasive and Mobile Computing}, 48:101--117, 2018.

\bibitem{DBLP:conf/mfcs/Ben-BasatEF18}
Ran Ben{-}Basat, Gil Einziger, and Roy Friedman.
\newblock Give me some slack: Efficient network measurements.
\newblock In {\em MFCS}, pages 34:1--34:16, 2018.

\bibitem{DBLP:conf/swat/Ben-BasatEFK16}
Ran Ben{-}Basat, Gil Einziger, Roy Friedman, and Yaron Kassner.
\newblock Efficient summing over sliding windows.
\newblock In {\em SWAT}, pages 11:1--11:14, 2016.

\bibitem{DBLP:conf/infocom/Ben-BasatEFK16}
Ran Ben{-}Basat, Gil Einziger, Roy Friedman, and Yaron Kassner.
\newblock Heavy hitters in streams and sliding windows.
\newblock In {\em IEEE INFOCOM}, pages 1--9, 2016.

\bibitem{DBLP:journals/corr/abs-1804-10740}
Ran Ben{-}Basat, Roy Friedman, and Rana Shahout.
\newblock Heavy hitters over interval queries.
\newblock {\em CoRR}, abs/1804.10740, 2018.

\bibitem{Clark:1996:EST:313852.314087}
David~R. Clark and J.~Ian Munro.
\newblock Efficient suffix trees on secondary storage.
\newblock In {\em SODA}, pages 383--391, 1996.

\bibitem{DBLP:journals/siamcomp/DatarGIM02}
Mayur Datar, Aristides Gionis, Piotr Indyk, and Rajeev Motwani.
\newblock Maintaining stream statistics over sliding windows.
\newblock {\em {SIAM} J. Comput.}, 31(6):1794--1813, 2002.

\bibitem{DBLP:conf/isaac/El-ZeinMN17}
Hicham El{-}Zein, J.~Ian Munro, and Yakov Nekrich.
\newblock Succinct color searching in one dimension.
\newblock In {\em ISAAC}, pages 30:1--30:11, 2017.

\bibitem{fusy2007estimating}
{\'E}ric Fusy and Fr{\'e}c{\'e}ric Giroire.
\newblock Estimating the number of active flows in a data stream over a sliding
  window.
\newblock In {\em ANALCO}, pages 223--231, 2007.

\bibitem{GibbonsT02}
Phillip~B. Gibbons and Srikanta Tirthapura.
\newblock Distributed streams algorithms for sliding windows.
\newblock In {\em {SPAA}}, pages 63--72, 2002.

\bibitem{Golynski:2006:ROL:1109557.1109599}
Alexander Golynski, J.~Ian Munro, and S.~Srinivasa Rao.
\newblock Rank/select operations on large alphabets: A tool for text indexing.
\newblock In {\em SODA}, pages 368--373, 2006.

\bibitem{DBLP:journals/algorithmica/GolynskiOR014}
Alexander Golynski, Alessio Orlandi, Rajeev Raman, and S.~Srinivasa Rao.
\newblock Optimal indexes for sparse bit vectors.
\newblock {\em Algorithmica}, 69(4):906--924, 2014.

\bibitem{DBLP:journals/tcs/HonSS11}
Wing{-}Kai Hon, Kunihiko Sadakane, and Wing{-}Kin Sung.
\newblock Succinct data structures for searchable partial sums with optimal
  worst-case performance.
\newblock {\em Theor. Comput. Sci.}, 412(39):5176--5186, 2011.

\bibitem{Jacobson:1988:SSD:915547}
Guy~Joseph Jacobson.
\newblock {\em Succinct Static Data Structures}.
\newblock PhD thesis, Pittsburgh, PA, USA, 1988.
\newblock AAI8918056.

\bibitem{DBLP:journals/cj/JoJORS17}
Seungbum Jo, Stelios Joannou, Daisuke Okanohara, Rajeev Raman, and
  Srinivasa~Rao Satti.
\newblock Compressed bit vectors based on variable-to-fixed encodings.
\newblock {\em Comput. J.}, 60(5):761--775, 2017.

\bibitem{miltersen-survey}
P.~B. Miltersen.
\newblock Cell probe complexity - a survey.
\newblock {\em FSTTCS}, 1999.

\bibitem{Munro:2001:SES:375519.375532}
J.Ian Munro, Venkatesh Raman, and S.Srinivasa Rao.
\newblock Space efficient suffix trees.
\newblock {\em J. Algorithms}, 39(2):205--222, 2001.

\bibitem{DBLP:conf/wea/NavarroP12}
Gonzalo Navarro and Eliana Providel.
\newblock Fast, small, simple rank/select on bitmaps.
\newblock In {\em SEA}, pages 295--306, 2012.

\bibitem{DBLP:conf/alenex/OkanoharaS07}
Daisuke Okanohara and Kunihiko Sadakane.
\newblock Practical entropy-compressed rank/select dictionary.
\newblock In {\em ALENEX}, pages 60--70, 2007.

\bibitem{Patrascu:2006:TTP:1132516.1132551}
Mihai P\u{a}tra\c{s}cu and Mikkel Thorup.
\newblock Time-space trade-offs for predecessor search.
\newblock In {\em ACM STOC}, pages 232--240, 2006.

\bibitem{DBLP:conf/wads/RamanRR01}
Rajeev Raman, Venkatesh Raman, and S.~Srinivasa Rao.
\newblock Succinct dynamic data structures.
\newblock In {\em WADS}, pages 426--437, 2001.

\bibitem{DBLP:journals/talg/RamanRS07}
Rajeev Raman, Venkatesh Raman, and Srinivasa~Rao Satti.
\newblock Succinct indexable dictionaries with applications to encoding
  \emph{k}-ary trees, prefix sums and multisets.
\newblock {\em {ACM} Trans. Algorithms}, 3(4):43, 2007.

\end{thebibliography}
\ifdefined\fullversion
\newpage
\appendix
\section{Proof of Theorem~\ref{thm:ldselect}}\label{app:ldselect}

\begin{theorem*}
Any data structures that supports $\vrankA_1$ or $\dselectA_1$ queries with additive error $\delta$ on a bit-string of length $n$ requires at least $\floor{n/2\delta}\lg{\delta}$ bits. 
\end{theorem*}
\begin{proof}
We first construct a set $V$ of bit-strings of length $n$ as follows.
We divide each bit-string $B$ into $\floor{n/2\delta}$ blocks $B_1$, $B_2$, \dots $B_{\floor{n/2\delta}}$ 
such that for $1 \le i < \floor{n/2\delta}$, $B_i = B[2\delta(i-1)+1 \dots 2\delta i]$ and 
$B_{\floor{n/2\delta}} = B[2\delta(\floor{n/2\delta}-1)+1 \dots n]$. 
Now for every $1 \le i \le \floor{n/2\delta}$, we set all bits in $B_i$ to $0$ if $i$ is odd.
If $i$ is even, we fill $B_i$ to $k \le \delta$ $1$'s followed $(\delta - k)$ $0$'s.
Thus there's only one choice of blocks $B_i$ (if $i$ is odd), and $\delta$ choices for blocks $B_i$ (if $i$ is even). Hence $|V| = \delta^{\floor{n/2\delta}}$.  	
Now consider two distinct bit-strings $B$ and $B'$ in $V$, and let $i$ be the even index of the leftmost block 
such that $B_{i} \neq B'_{i}$ and without loss of generality, $B_i$ and $B'_{i}$ has $k$ and $k'$ $1$s with $k < k'$ respectively.
Since for such block has $\delta$ zeros on both sides, it is easy to show that there is no value which is the answer of both $\vrankA_1((i-1)\delta+k', B, \delta)$ and $\vrankA_1((i-1)\delta+k', B', \delta)$ queries, and also there is no position in $B$ which is the answer of both $\dselectA_1(\ell, B, \delta)$ and $\vrankA_1(\ell, B', \delta)$ queries, where $\ell$ is number of $1$'s in $B'[1 \dots (i-1)\delta+k']$.
Thus any structure that supports either of these queries must distinguish between 
every element in $S$, and hence $\lg |V| = \floor{n/2\delta}\lg{\delta}$ bits are necessary to answer $\vrankA_1$ and $\dselectA_1$ queries.	
\end{proof}
\section{Proof of Theorem~\ref{thm:dselectLgLgLowerBound}}\label{app:dselectLgLgLowerBound}
\begin{theorem*}
	Any $((n/\delta) \lg^{O(1)}{\delta})$-bit data structure that supports $\dselectA_1$ 
	queries with an additive error $\delta = O(n^c)$, for some constant $0 < c \le 1$ on a 
	bit-string of length $n$ requires $\Omega(\lg{\lg{n}})$ query time.
\end{theorem*}
\begin{proof}
We reduce the predecessor search problem to the problem of supporting $\dselectA_1$ queries.
Given a set $S \subseteq \{1, \dots, n \}$, a predecessor query, $\predec(i,S)$, for $1 \le i \le n$, 
returns the largest elements in $S$ that is smaller than $i$.
Let $S \subseteq \{1, \dots, n \}$ be a given set on which we want to support $\predec$ queries, with
$|S| = n/\delta$, where $\delta = O(n^c)$, for some constant $0 < c \le 1$. For this range 
of parameters, Patrascu and Thorup~\cite{Patrascu:2006:TTP:1132516.1132551} showed that 
any data structure that represents $S$ using $O(n \lg^{O(1)} n)$ bits needs $\Omega(\lg\lg n)$ time to support 
$\predec$ queries. We now show that any data structure that supports $\dselectA_1$ queries can
be used to obtain a data structure that supports $\predec$ queries, using asymptotically the same 
space and query time. The theorem immediately follows from this reduction.

Let $S$ be a given set. Let $S' = \{k\delta | 1 \le k \le \floor{n/\delta} \} \cup \{ n \}$.
We call the elements in $S'$ as the {\em dummy elements}. 
Let $S_1 = S \cup S'$, and let $x_1, x_2, \dots, x_\ell = n$
be the elements of $S_1$ in sorted order, for some $n/\delta \le \ell \le 2n/ \delta$ (note that both $S$ and $S'$ have size $n/\delta$).
The dummy elements in $S_1$ ensure that $x_1 \le \delta$, and $x_i - x_{i-1} \le \delta$, for $1 < i \le \ell$.
Now, consider the bit-string $B = B_1 B_2 \dots B_\ell$, where block $B_1 = 0^{2 \delta - x_1} 1^{x_1}$, and
for $1 < i \le \ell$, block $B_i = 0^{2 \delta - x_i + x_{i-1}} 1^{x_i - x_{i-1}}$ (i.e., $B$ encodes the differences between 
successive elements of $S_1$ using fixed-length right-justified unary codes of size $2 \delta$). Note that $B$ contains $x_\ell = n$ 1's, and 
has length $2 \delta \ell \le 2n$.
In addition, we store an array $A$ of length $\ell$ such that $A[i] = \predec(x_i,S)$, which uses $O((n/\delta) \lg n)$ bits.

Suppose that there is a data structure $X$ that uses $s(n,\delta)$ space, and supports $\dselectA_1$ queries on $B$ in $t(n,\delta)$ time. 
To answer the query $\predec(x,S)$, we first perform the $\dselectA_1(x, B, \delta)$ on $X$. Let $B_i$ be the block to which this answer 
belongs. Since each block starts with a sequence of at least $\delta$ zeros, and since $\dselectA_1(x, B, \delta) \le \select_1(x,B)$,
it follows that $x_i \le x < x_{i+1}$. Hence we return $A[x_i]$ as the answer of $\pred(x)$. 
Thus, from the assumption about the data structure $X$, we can obtain a structure that uses $s(n,\delta)+ O((n/\delta)\lg n)$ bits and supports $\predec$
queries in $t(n,\delta)+O(1)$ time. The theorem follows from this reduction, and the predecessor lower bound mentioned above.
\end{proof}

\section{Proof of Theorem~\ref{thm:ubds}}\label{app:ubds}

\begin{theorem*}
For a bit-string $B$ of length $n$, there is a data structure of size $(n/\delta)\lg{\delta}+o((n/\delta)\lg{\delta})$ bits, which supports $\vrankA_1$ queries on $B$ using $O(1)$ time and $\dselectA_1{}$ queries on $B$ using $SPS(n/\delta, n)$ time. 
\end{theorem*}
\begin{proof}
We divide the $B$ into $\ceil{n/\delta}$ blocks $B_1$, $B_2$, \dots $B_{\ceil{n/\delta}}$, defined exactly as in the proof of  Theorem~\ref{thm:ub1}. Now we define an array $C[1 \dots \ceil{n/\delta}]$ of length $\ceil{n/\delta}$ such that for $1 \le i \le \ceil{n/\delta}$, $C[i]$ is the number of $1$'s in $B_i$.
We represent the array $C$ using the structure of Lemma~\ref{lem:prefix}, to support $\summ{}$ and $\search{}$ queries on $C$, using $O((n/\delta)\lg{\delta})$ bits.
One can easily show that $\vrankA_1(j, B, \delta)$ query is same as the answer of $\summ{}(\floor{j/\delta})$ query on $C$, which can be answered in $O(1)$ time by Lemma~\ref{lem:prefix}. Also it is easy to show that 
$\rank_1{}(j-\delta, B) <  \summ{}(\floor{j/\delta}) \le \rank_1{}(j, B)$.
To answer the query $\dselectA_1(j, B, \delta)$, we first find the block $B_i$ in $B$ which contains the position $\select_1(j,B)$, using $i = \search{}(j)$ on $C$, and return $(i-1)\delta$ as answer to $\dselectA_1(j, B, \delta)$. It is easy to see that $\select_1(j,B) - \delta < (i-1) \delta \le \select_1(j,B)$.
\end{proof}
\section{Proof of Theorem~\ref{thm:largeAlphabetLB}}
\begin{theorem*}\label{app:largeAlphabetLB}
Given a multiset $S$ where each element is from the universe $U =\{1, 2, \dots , n\}$ of size $n$, any data structure that supports $\rankA$ queries on $S$ requires at least $\floor{n/\ceil{\delta/\ell}}\lg{(\max{(\floor{\ell/\delta},1)}+1)}$ bits, where $\ell$ is a bound on the  maximum frequency of each element in $S$.
\end{theorem*}
\begin{proof}
Note that $S$ can be represented by a sequence $S_1, S_2 \dots S_n$ of size $n$, where $S_i \le \ell$ denotes a frequency of $i$ in $S$. 
Now we first set $\mu = {\delta/\ell}$ and denote 
$I$ as $\{\min{(\delta k, \ell)} | k \in \{0, 1, \dots \max{(\floor{1/\mu}, 1)}\} \subset \{0, 1, \dots \ell\} \}$, 
and denote $\bar{I}$ as $\{\sigma^{\ceil\mu} | \sigma \in I\}$.
Next, consider all inputs that contains a sequence of 
$\floor{n/\ceil{\mu}}$ \textit{blocks} padded by zeros, such that each block is a member of $\bar I$; that is, consider 
$\mathcal I =  \bar I^{\floor{n/\ceil{\mu}}}\cdot 0^{n-(n\mod \ceil{\mu})}$. It is easy to show that every input of $\mathcal{I}$ gives a representation of $S$. We show that every two distinct inputs in $\mathcal I$ must lead to distinct answer of a $\rankA$ query, thereby implying a $\ceil{\lg |\mathcal I|}$ bits lower bound as required. 
Let  two distinct set $S_1$ and $S_2$ be represented by the sequences in $\mathcal{I}$ such as 
$x_1=x_{1,1}x_{1,2}\cdots x_{1,\floor{n/\ceil{\mu}}}0^{n-(n\mod \ceil{\mu})}$ and 
$x_2=x_{2,1}x_{2,2}\cdots x_{1,\floor{n/\ceil{\mu}}}0^{n-(n\mod \ceil{\mu})}$ respectively such that $x_{\alpha,\beta}\in \bar I$ for any $\alpha\in\{1,2\},\beta\in\{1,\ldots,\floor{n/\ceil{\mu}}\}$.
Also let $t$ be a leftmost index such that $x_{1, t} \neq x_{2, t}$. 
Now we consider $\rankA(\ceil{\mu}t, S_1)$ and $\rankA(\ceil{\mu} t, S_2)$ queries. If $\mu \le 1$, then $\floor{n/\ceil{\mu}}=n$ and (due to the definition of $I$) $|x_{1,t}-x_{2,t}|\ge \delta$, which implies that there is no answer which satisfies both $\rankA(\ceil{\mu} t, S_1)$ and $\rankA{}(\ceil{\mu}t, S_2)$ queries.
On the other hand, $\mu > 1$ means that $I = \{0,\ell\}$ and thus either $x_{1,t}=0^{\ceil{\mu}}, x_{2,t}=\ell^{\ceil{\mu}}$ or $x_{1,t}=\ell^{\ceil{\mu}}, x_{2,t}=0^{\ceil{\mu}}$. In either case,
$|(\rankA{}(\ceil{\mu} t, S_1) - \rankA(\ceil{\mu} t, S_2)| \ge \delta$. We established that if two inputs in $\mathcal I$ lead to the same configuration of $\rankA{}$ queries, the error for one of them would be at least $\delta$ while we assumed it is strictly lower.
\end{proof}
\section{\rankA{} and \vselectA{} queries on multiset $S$ when the frequency 
of each elements in $S$ is at most $\ell$}
\label{sec:multisetupper}
\begin{itemize}
\item {\bf Case 1. \bm{$\delta \le \ell$}:}
In this case, we first observe that $|S| \le n \ell$. 
Hence, $B_S$ is a bit-string with $n$ $0$'s and at most $n \ell$ $1$'s,
and $B'_S$ has $n$ $0$'s and at most $n \ell/\delta$ $1$'s.
To support $\rankA$ on $S$, we need to support $\select_0$ on $B'_S$.
We represent the bit-wise complement of $B'_S$ using the structure of Lemma~\ref{lem:RRR}(a), 
which takes $\mathcal{B}(n+\floor{n \ell/\delta}, \floor{n \ell/\delta})+o(n)$ bits
and supports $\select_0$ on $B'_S$ in $O(1)$ time.
Using this structure, we can achieve optimal space usage, and support $\rankA$ queries on $S$ in $O(1)$ time.
Alternatively, we can represent $B'_S$ using the structure of Lemma~\ref{lem:RRR}(b),
which takes $\mathcal{B}(n+\floor{n \ell/\delta}, \floor{n \ell/\delta})+o(n+\floor{n \ell/\delta})$ bits,
and supports $\rank_0$, $\rank_1$, $\select_0$ and $\select_1$ queries on $B'_S$ in $O(1)$ time.
Using this structure, we can support both $\rankA$ and $\vselectA$ queries on $S$ in $O(1)$ time,
while using asymptotically optimal space when $\ell=\Theta(\delta)$.
\item {\bf Case 2. \bm{$\delta > \ell$}:}
In this case, we first set $\mu = \floor{\delta/\ell}$, and define a bit-string
$B'[1 \dots \ceil{n/\mu}]$ of length $\ceil{n/\mu}$ such that $B'[i] = 1$ if and only if 
there exists a $1$ between the positions of the $(i-1)\mu$-th $0$ and the $(\min(i\mu, n))$-th $0$ 
in $B'_S$. Since $\mu\ell \le \delta$, there exists at most a single $1$ between these two positions.
Now, using Lemma~\ref{lem:clark}, we construct a $n/\mu+o(n/\mu) = n\ell/\delta+o(n \ell/\delta)$-bit data structure which supports $\rank{}_1$ and $\select{}_1$ queries on $B'$ in constant time. 
Then one can show that $\delta(\rank{}_1(\floor{i/\mu}, B'))+\ell(i \mod \mu)B'[\ceil{i/\mu}]$ is an answer to the query $\rankA(i, S, \delta)$, using an argument similar to the one in the proof of Theorem~\ref{thm:ub1}.
For $\vselectA(i, S, \delta)$ queries, we set $\mu = \floor{\delta/2\ell}$ and construct a same structure as above, using $2n\ell/\delta+o(n \ell/\delta)$ bits.  
Since there are at most $((\mu+(i \mod \mu))\ell < \delta$ elements, 
we can answer $\vselectA$ query in constant time by returning $\mu(\select{}_1(\floor{i/\delta}, B')-1)$.
Therefore, our data structure supports $\rankA$ queries in constant time with optimal space, and twice the optimal space for supporting both $\rankA$ and $\vselectA$ queries in constant time (note that at least $\floor{n/\ceil{\delta/\ell}}$ bits are necessary in this case).
\end{itemize}
\section{Proof of Theorem~\ref{thm:large}}\label{apx:generalstaticproof}
\begin{theorem*}
Let $A$ be a string of length $n$ over the alphabet $\Sigma=\{1,2, \dots, \sigma\}$. Then for any $1 \le j \le \sigma$, 
one can support $\rankA{}_j$ and $\vselectA{}_j$ queries in $O(\lg{\lg{\sigma}})$ time
using  $2n/\delta\lg{(\sigma+1)}+o((n/\delta)\lg{(\sigma+1)})$ bits.
\end{theorem*}
\begin{proof}
We first 
divide the string $A$ into $\ceil{n/\delta}$ 
blocks $A_1$, $A_2$ \dots $A_{\ceil{n/\delta}}$ such that for $1 \le i < \ceil{n/\delta}$, 
$A_i = A[\delta(i-1)+1 \dots \delta i]$ and $A_{\ceil{n/\delta}} = A[\delta(\ceil{n/\delta}-1)+1 \dots n]$.
Then we construct a new string $A' = A_1\$A_2\$ \dots \$A_{\ceil{n/\delta}}\$$ of length $n+\ceil{n/\delta}$,
where $\$$ is a symbol not in $\Sigma$.
Now we construct yet another string $A''$ of length at most $\floor{n/\delta}+\ceil{n/\delta}$, 
which is a subsequence of $A'$, obtained by only keeping every $i \delta$-th occurrence of all 
the symbols from $\Sigma$, for $1 \dots i \le \floor{n/\delta}$ in $A'$, and also all the 
occurrences of $\$$ in $A'$, while removing all the other characters in $A'$.
We then represent $A''$ using the structure of Lemma~\ref{lem:gmr}, which takes $(2n/\delta)\lg{(\sigma+1)}+o((n/\delta)\lg{(\sigma+1)})$ bits, and supports $\rank{}$ and $\select{}$ queries on $A''$ in $O(\lg{\lg{(\sigma+1)}})$ and $O(1)$ time, respectively. 

For answering the $\rankA{}_j(i, A, \delta)$ query, we first compute the position, $b_i$, of the $\floor{i/\delta}$-th $\$$ in $A''$, 
in constant time, using $b_i = \select{}_{\$}(\floor{i/\delta}, A'')$.
Then by an argument similar to the one in the proof of Theorem~\ref{thm:ub1}, one can 
show that $\delta\rank{}_{j}(b_i, A'')+(i \mod \delta)\kappa_{j}(A[\ceil{i/\delta}])$ gives an answer of the $\rankA{}_j(i, A, \delta)$ query, where $\kappa_{j}(A[i])$ is an indicator function which defined as $\kappa_{j}(A[i])= 1$ if $A[i] = j$, and 0 otherwise.
Thus, $\rankA{}_j(i, A, \delta)$ query can be answered in $O(\lg{\lg{\sigma}})$ time.
Similarly, it is easy to see that we can answer the $\vselectA{}_j(i, A, \delta)$ query in 
$O(\lg{\lg{\sigma}})$ time by returning $\delta\rank_{\$}(\select_{j}(\floor{i/\delta}-1, A''), A'')+(i \mod d)$.
\end{proof}
\section{Succinct data structure for answering \suffixsum{} and \iss{} queries}
\label{subsec:nonapprox}
In this section, we describe an $n+o(n)$-bit data structure for answering $\suffixsum{}$ and $\iss{}$ queries on a binary steam in constant time while supporting constant time updates.
Our data structure is based on the data structure of Lemma~\ref{lem:clark} for answering $\rank{}$ and $\select{}$ queries on a bit-string. 
We first consider the stream as the stream of \textit{frames}, 
which is $n$-bit consecutive elements of the stream.
Our main goal is to maintain the size of the data structure at most $n+o(n)$ bits while 
answering $\suffixsum{}$ and $\iss{}$ queries on the stream whose answer (or range) covers both current and previous frames.
The overall idea for achieving the goal is as follows. 
When the element of new frame arrives we construct a data structure of Lemma~\ref{lem:clark} over the new frame,  
which replaces the data structure constructed over the oldest element in the previous frame.
%
Now we describe the details as follows. 
We first store the last $n$ elements of the stream into a
circular array $\mathcal{W}[1 \dots n]$ of size $n$
such that $\mathcal{W}[i]$ stores the $i$-th leftmost element in the frame.
Also, for $1 \le t \le n$, let $f_t$ (resp., $f'_t$) be the $t$-th arrived element in the current (resp., previous) frame 
\\\\
\noindent\textbf{i) $\suffixsum{} (i, n)$ queries :} 
For answering $\suffixsum{}$ queries,
We divide a frame into $\ceil{n/\lg^2{n}}$ blocks of size $\lg^2{n}$, 
and divide each block again into $2\lg{n}$ sub-blocks of size $\lg{n}/2$. 
At the end of the each block we store the number of 1's from the beginning of the current frame into an array $\mathcal{C}$ of size $\ceil{n/\lg^2{n}}$. 
Similarly at the end of the each sub-block
, we store the number of 1's from the beginning of the current block into an array $\mathcal{SC}$ of size $\ceil{2n/\lg{n}}$.
To update both arrays in constant time, we maintain two counters $c$ and $sc$, which count the number of 1's in the current frame and the current block respectively.
Also, we construct a look-up table $T$ such that for any string $s \le \lg{n}/2$ 
and an index $1 \le i \le |s|$, $T[s][i]$ stores the number of 1's in the suffix of $s$ of size $i$.

Now we describe how to update the data structure in constant time when $f_t$ arrives from the stream.
We first update $\mathcal{W}[t]$ to $f_t$, increase $c$ and $sc$ by 1 if $f_t = 1$, and update auxiliary structures as follows.
\begin{itemize}
\item {Case 1. ${(t \mod \lg{n}/2}) = 0$ :}
Since $t$ is the rightmost position in the sub-block, we set $\mathcal{SC}[2t/\lg{n}] = sc$. 
\item {Case 2. ${(t \mod \lg^2{n}}) = 0$ :}
Since $t$ is the rightmost position in the block,
we set $\mathcal{SC}[\ceil{2t/\lg{n}}] = sc$, $\mathcal{C}[t/\lg^2{n}] = c$, and reset $sc$ to zero. 
\item {Case 3. $t = n$ :}
Since $t$ is the rightmost position in the current frame, 
we set $\mathcal{SC}[\ceil{2t/\lg{n}}] = sc$, $\mathcal{C}[\ceil{t/\lg^2{n}}] = c$, 
and reset $sc$ and $c$ to zero. 
\end{itemize}

By the procedure described above, it is clear that 
whenever the new element arrives from the stream, 
we can update the data structure in constant time. 
Now we consider how to answer $\suffixsum{} (i, n)$ query after $f_t$ arrives.
If $i \le t$, it is enough to count the number of 1's in the current frame. 
In this case, we first count the number of 1's
in $f_{i'}, \dots f_t$
in constant time by returning $c - (\mathcal{C}[\floor{(t-i+1)/\lg^2{n}}]+\mathcal{SC}[\ceil{2(t-i+1)/\lg{n}}])$
, where $i'$ is the rightmost position of the sub-block which contains $f_{t-i+1}$. 
Also the number of 1's in $f_i \dots f_{i'}$ is
$T[f_{i'-\lg{n}/2} \dots f_{i'}][i'-i+1]$, which can be answered in constant time.
By adding these two values, we can answer $\suffixsum{} (i, n)$ query in constant time.
If $i > t$, the query range covers both current and previous frame. 
In this case, the answer of $\suffixsum{} (i, n)$ query is $c$+(number of 1's in $f'_{n-(i-t-1)} \dots f'_n$), 
which can be answered in constant time by using a similar argument as above.
Finally for space usage, we can store an array $\mathcal{W}$ and counters using $n+O(\lg{n}) = n+o(n)$ bits. 
Also we can store $\mathcal{C}$ using $O(n\lg{n}/ \lg^2{n}) = o(n)$ bits, $\mathcal{SC}$ using $O(n\lg{\lg{n}}/ \lg{n}) = o(n)$ bits, and $T$ using $O(2^{\lg{n}/2}\lg{n}\lg{\lg{n}}) = o(n)$ bits.
Therefore, the total space of the data structure is $n+o(n)$ bits.
\\\\
\noindent\textbf{ii)  $\iss{} (i, n)$ queries :} 
For answering $\iss{}$ queries, 
we first mark  the positions of 
$f_1$, $f_n$, and every $\lg{n}\lg\lg{n}$-th 1's 
in the frame $f = f_1 \dots f_n$ and define a block of $f$ as a sub-string between two marked positions.
For block $C = c_1 \dots c_{|C|}$ in $f$, if the size of $C$ is greater than $\lg^2{n}(\lg\lg{n})^2$,  
we store all the positions of 1's in $C$
into an array $\mathcal{C}$ of size $n/\lg^2{n}(\lg\lg{n})^2 \times \lg{n}\lg\lg{n} = n/\lg{n}\lg\lg{n}$.
If the size of $C$ is less than $\lg^2{n}(\lg\lg{n})^2$,
we mark the positions $c_1$, $c_{|C|}$, and every $(\lg\lg{n})^2$-th 1's in $C$ and
define a sub-block $SC$ of $C$ as a sub-string between two marked positions in $C$.
Now for each sub-block $SC$, if the size of $SC$ is greater than $(\lg\lg{n})^4$,  
we store all the positions of 1's in $SC$
into an array $\mathcal{SC}$ of size $n/(\lg\lg{n})^4 \times (\lg\lg{n})^2 = n/(\lg\lg{n})^2$.
If not, we answer $\iss{}$ queries in $SC$ using a look-up table $T'$.
For all possible bit-string $s$ of size $(\lg\lg{n})^4$ and $i \le (\lg\lg{n})^4$, 
$T'[s][i]$ stores a position of the $i$-th rightmost 1 in $s$.
Also all the marked positions in frames (resp., blocks) are stored into an array 
$M$(resp., $M_s$) of size $n/\lg{n}\lg\lg{n}$ (resp., $n/(\lg\lg{n})^2$).

If $M[m]-M[m-1] > \lg^2{n}(\lg\lg{n})^2$ (In the rest of this section, let $A[i]$ imply $A[(i \mod |A|)+1]$), 
we store a pointer in $M[m]$ which indicates the position of $\mathcal{C}$ 
that contains the position of the first 1 in $f_{M[m-1]+1} \dots f_{M[m]}$.
If not, we store a pointer in $M_s$ 
which indicates the first marked position  in $f_{M[m-1]+1} \dots f_{M[m]}$.
Similarly for each element in $M_s[m_s]$, we maintain a pointer which refers to the appropriate position of $\mathcal{SC}$ if
$M_s[m_s] - M_s[m_s-1] \le (\lg\lg{n})^4$.
We also maintain four indices $c$, $sc$, $m$, and $m_s$ such that $\mathcal{C}[c]$, $\mathcal{SC}[sc]$, $M[m]$, and $M_s[m_s]$ are the last-updated values respectively, and $\mu$ which stores the number of 1's in the current frame.
Finally for $c$, $sc$, and $m_s$, we maintain their copies $c'$, $sc'$, and $m_s'$ respectively, which are initially identical to their original values.

Now we describe how to update the data structure in constant time when $f_t$ arrives from the stream.
We first update $\mathcal{W}[i]$ to $f_t$ and
if $f_t  = 1$, we i) increase $\mu$ by 1,
ii) update $c'$ and $sc'$ to be $(c'+1) \mod n/\lg{n}\lg\lg{n}$ and $(sc'+1) \mod n/(\lg\lg{n})^2$ respectively, 
and iii) set $\mathcal{C}[c'] = t$ and $\mathcal{SC}[sc'] = t - M[m]+1$.
Next based on the $c'$ and $sc'$, we update auxiliary structures as follows.

\begin{itemize}
\item {Case 1. $((sc' - sc) \mod (\lg\lg{n})^2) = 0$ :} 
We first increase $m'_s$ by 1 (in modulo $n/(\lg\lg{n})^2$) and set $M_s[m'_s] = t-M[m]+1$. 
If $t - M_s[m_s-1]+1 < (\lg\lg{n})^4$, 
we reset $sc'$ to $sc$ since we do not store the position of 1's explicitly in this sub-block.
Also, if $t-M[m]+1 > \lg^2{n}(\lg\lg{n})^2$, we reset both $sc'$ and $m'_s$ to $sc$ and $m_s$
respectively, which implies there is no marked position in the block that contains $f_t$.
Otherwise, we store a pointer $M_s[m'_s]$ to $\mathcal{SC}[(sc'-(\lg\lg{n})^2+1)]$.
\item {Case 2. $((c'-c) \mod \lg{n}\lg{\lg{n}}) = 0$ or $t=n$ :}
We first increase $m$ and $m_s$ by 1 (in modulo $n/(\lg\lg{n})^2$ and $n/\lg{n}\lg{\lg{n}}$ respectively), 
and set $M[m] = t$ and $M_s[m_s] = t-M[m-1 ]+1$.
If $t-M[m-1 ]+1  > \lg^2{n}(\lg\lg{n})^2$,
we reset both $sc'$ and $m'_s$ to $sc$ and $m_s$ respectively, and 
store a pointer $M[m]$ to $\mathcal{C}[c'-\lg{n}\lg\lg{n}+1]$.
If not, we first reset $c'$ to $c$ and store a pointer  $M[m]$ to $M_s[m_s - \lg{n}/\lg{\lg{n}}+1]$ and consider two cases as
i) store a pointer $M_s[m_s]$ to $\mathcal{SC}[sc'-(\lg\lg{n})^2+1]$
if $i - M_s[m_s-1]+1>(\lg\lg{n})^4$, and ii) reset $sc'$ to $sc$ otherwise.
Finally, we update $c$, $sc$, and $m_s$ to be $c'$ $cs'$, and $m'_s$ respectively.
In addition to that when $t=n$, we copy $m$ and $\mu$ into 
another indices $ind$ and $total$ respectively using $\lg{n}$ bits, and reset $\mu = 0$.
\end{itemize}

By the procedure described above, it is clear that 
whenever the new element arrives from the stream, 
we can update the data structure in constant time. 
Now we describe how to answer $\iss{} (i, n)$ query after $f_t$ arrives from the stream. 
We first check whether $i \le \mu$ or not. 
If $i \le \mu$, the answer is in the current frame and we consider the following cases.
\begin{itemize}
\item{$i \le c'-c$ :}
In this case, we return $n-\mathcal{C}[c'-i+1]$.
\item{$i > c'-c$ :}
In this case, let $i' = (m-\ceil{(i-(c'-c))/\lg{n}\lg{\lg{n}}}+1) \mod n/\lg{n}\lg\lg{n}$. 
If the pointer in  $M[i']$ indicates $\mathcal{C}[i'']$, return
$n-\mathcal{C}[i''+d]$, where $d = \lg{n}\lg{\lg{n}}-(i-i'\lg{n}\lg{\lg{n}})$.
If the pointer in $M[i']$ indicates $M_s[i'_s]$, let
$j = (i'_s+\floor{ d/(\lg\lg{n})^2}) \mod n/(\lg\lg{n})^2$.
If the pointer $M_s[j]$ indicates $\mathcal{SC}[j']$, return
$n-(M[i']+\mathcal{SC}[j'+(i-(\lg\lg{n})^2\floor{i/(\lg\lg{n})^2})])$, 
and otherwise return $n-(M[i']+M_s[j-1]+T'[s][d'])$, where
$s=f'_{M[i']+M_s[j-1]} \dots f'_{M[i']+M_s[j]-1}$ and $d' = (\lg\lg{n})^2-(\lg\lg{n})^2\floor{(i-(c'-c))/(\lg\lg{n})^2}$.
\end{itemize}

If $i > \mu$, the answer of $\iss{} (i, n)$ query is in the position of the previous frame.
In this case, the answer of $\iss{} (i, n)$ query is same as the 
position of $(i-\mu)$-th rightmost one
in the previous frame, which can be answered by the similar argument as the above procedure using $ind$ and $total$.
Therefore, we can answer $\iss{} (i, n)$ queries in constant time. 
For space usage, $\mathcal{C}$ and $M$ takes $n\lg{n}/\lg{n}\lg{\lg{n}} = o(n)$ bits, and 
$\mathcal{SC}$ and $M_s$ takes $n\lg{\lg{n}}/(\lg\lg{n})^2 = o(n)$ bits.
Also we can store $T'$ using $O(2^{(\lg\lg{n})^4}(\lg\lg{n})^4\lg{\lg{\lg{n}}}) = o(n)$ bits, and 
it is clear that we can store all other counters and pointers using at most
$o(n)$ bits. Therefore, we can answer  $\iss{} (i, n)$ queries in constant time, using $n+o(n)$ bits of space.

\section{Proof of Lemma~\ref{lem:slidingQsrSpaceProof}}\label{apx:slidingQsrSpaceProof}
\begin{lemma*}
	Algorithm~\ref{alg:approxSliding} requires $\sFactor\cdot\floor{n/\max{(\floor{\mu},1)}}\cdot\log\big({\ceil{\mu^{-1}} + 1}\big) + O\parentheses{\log n}$~bits.
\end{lemma*}
\begin{proof}
The algorithm utilizes three variables: $\ranker$ that requires $\sFactor\cdot s\logp{z+1}$ bits, $\mathfrak{r}$ that uses $O(\mathfrak{b}\log\nu)$ space, and $o$ that is allocated with $\ceil{\log n}$ bits. 
Recall that $s= \UBnBlocks$ is the number of blocks that can overlap with the maximal $n$-sized window and $z\triangleq\floor{\mu^{-1}\nu}$ is a bound on $\newInputLetter$.
Overall, the number of bits used by our construction is
\begin{align*}
&\sFactor\cdot s\logp{z+1} + O(\mathfrak{b}\log\nu) + \ceil{\log n} \\
=& \sFactor\cdot \UBnBlocks\logp{\floor{\mu^{-1}\nu+1}+1} + O({\ceil{\logp{n/\mu}+\log\log n}}\log\nu) + O\parentheses{\log n}.
\end{align*}
Since $\nu=\max{(\floor{\mu\cdot\sNegFactor},1)}$, we get the desired bound.
\end{proof}
\section{Proof of Theorem~\ref{thm:slidingCorrecntess}}\label{app:slidingCorrecntess}
\begin{theorem*}
	Algorithm~\ref{alg:approxSliding} solves ssA while processing elements and answering queries in constant time.
\end{theorem*}
\begin{proof}
\newcommand{\lastT}{h}
For the proof, we define a few quantities that we also use in our {\sc Query} procedure: $\numBits \triangleq \ceil{\frac{i-o}{\nu}}, \setBits \triangleq \ranker.\mbox{\sc Query}\parentheses{\numBits}, \lastBit \triangleq \setBits \allowbreak-\break \ranker.\mbox{\sc Query}\parentheses{\numBits-1}$ and $\outBits \triangleq \outBitsVal$. 
We assume that the index of the most recent element is 
$\lastT\triangleq \outBits + i,$ such that $x_1$ is the first element in the chunk of $\lastBit$ and $o=\parentheses{\lastBit \mod \nu}$ is the offset within the current chunk. We also denote $\lastBlockEnd\triangleq \lastT-o$, such that $x_\lastBlockEnd$ is the last element of the most recently completed chunk.
Figure~\ref{fig:slidingRankerQueryProof} illustrates the setting.
By the correctness of the $\ranker$ exact suffix sum algorithm, and as illustrated in Figure~\ref{fig:slidingRankerQueryProof}, we have that \setBits{} is the sum of the last \numBits{} added to $\ranker$, that \lastBit{} is the value of the element that represents the last chunk that overlaps with the queried window. Also, notice that \outBits{} is the number of elements in that chunk that are not a part of the  window.
\begin{figure}[]
	\centering
	\includegraphics[width=\linewidth]{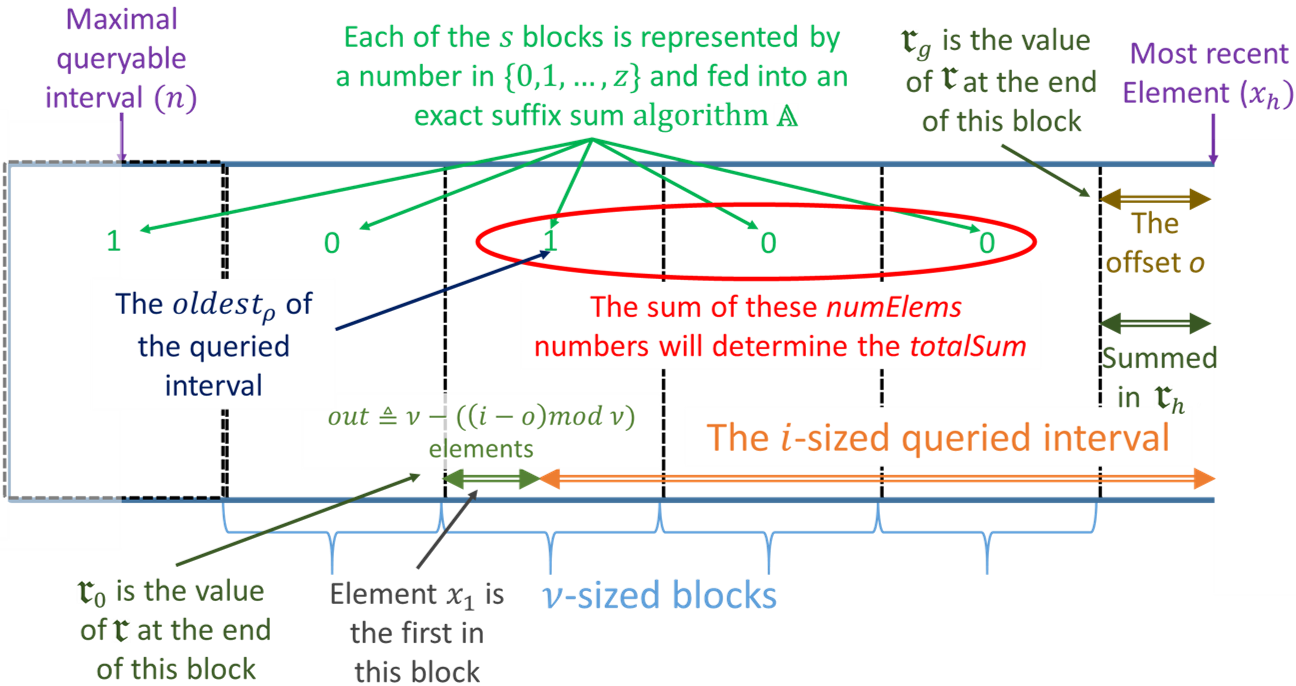}
	\caption{Theorem~\ref{thm:slidingCorrecntess} proof's setting, with all relevant quantities that Algorithm~\ref{alg:approxSliding} uses illustrated.}
	\label{fig:slidingRankerQueryProof}
\end{figure}
For any $t\in \mathbb N$, we denote by $\mathfrak{r_t}$ the value of $\mathfrak{r}$ \emph{after} the $t^{th}$ item was added; e.g., $\mathfrak{r_{\lastT}}$ is the value of $\mathfrak{r}$ at the time of the query and $\mathfrak{r_\lastBlockEnd}$ is its value before the current chunk. Notice that $\mathfrak{r_0}$ is also at the end of a chunk (that does not overlap with the queried interval).  For other variables, we consider their value at query time.

When a chunk ends (Line~\ref{line:end-of-block}), we effectively perform $\mathfrak{r}\gets \mathfrak{r}\mod \sensitivity$ (lines~\ref{line:setBit} and \ref{line:reduceS}), thus:
\begin{align}
0 \le \mathfrak{r_0} \le \sensitivity - 1.\label{eq:s0}
\end{align}
Our goal is to estimate the quantity 
\begin{align}
S_i \triangleq \sum_{d=\lastT-i+1}^{\lastT} x_d = \sum_{d=\outBits+1}^{\lastT} x_d.\label{eq0}
\end{align}
Recall that our estimation (Line~\ref{line:est}) is:
\begin{multline}
\widehat{S_i}\triangleq\mathfrak{r_{\lastT}} - \parentheses{\sensitivity - 1/2}+\sensitivity\cdot\setBits -\ell\cdot\lastBit\cdot\outBits\\
= \mathfrak{r_{\lastBlockEnd}} + \sum_{d=\lastBlockEnd+1}^{\lastT} Round_{\nBits}(x_d) - \parentheses{\sensitivity - 1/2}+\sensitivity\cdot\setBits -\ell\cdot\lastBit\cdot\outBits,
\label{eq1}
\end{multline}
where the last equality follows from the fact that within a chunk we simply sum the rounded values (Line~\ref{line:sum}).
Next, observe that we sum the rounded values in each chunk and that if $\mathfrak{r}$ is decreased by $k\cdot \sensitivity$ (for some $k\in\mathbb N$) at Line~\ref{line:reduceS}, then we set one of the last $\numBits$ elements added to $\ranker$ to $k$. This means that:
\begin{align}
\mathfrak{r_0} + \sum_{d=1}^{\lastBlockEnd} Round_{\nBits}(x_d) 
= \mathfrak{r_{\lastBlockEnd}} + \sensitivity\cdot\ranker.\mbox{\sc Query}\parentheses{\numBits} = \mathfrak{r_{\lastBlockEnd}} + \sensitivity\cdot\setBits.\label{eq2}
\end{align}
Plugging~\eqref{eq2} into~\eqref{eq1} gives us
\begin{multline}
\widehat{S_i}= \mathfrak{r_0} + \sum_{d=1}^{\lastBlockEnd} Round_{\nBits}(x_d)  + \sum_{d=\lastBlockEnd+1}^{\lastT} Round_{\nBits}(x_d) - \parentheses{\sensitivity - 1/2} -\ell\cdot\lastBit\cdot\outBits.
\label{eq3}
\end{multline}
Joining~\eqref{eq3} with \eqref{eq0}, we can express the algorithm's error as:
\begin{multline}
\widehat{S_i} - S_i = \mathfrak{r_0} + \sum_{d=1}^{\outBits} Round_{\nBits}(x_d)  + \sum_{d=\outBits+1}^{\lastT} \biggParentheses{Round_{\nBits}(x_d) - x_d} - \parentheses{\sensitivity - 1/2} -\ell\cdot\lastBit\cdot\outBits\\
= \mathfrak{r_0} + \sum_{d=1}^{\outBits} Round_{\nBits}(x_d)  + \xi - \parentheses{\sensitivity - 1/2} -\ell\cdot\lastBit\cdot\outBits
,\label{eq4}
\end{multline}
where $\xi$ is the rounding error, defined as
$
\xi \triangleq \sum_{d=\outBits+1}^{\lastT} \biggParentheses{Round_{\nBits}(x_d) - x_d}.
$

Since each rounding of an integer $x\in\frange{\ell}$ has an error of at most $\frac{\ell}{2^{\nBits}}$, and as
we round $i\le n$ elements, we have that the rounding error satisfies
\begin{align}
0 \ge \xi \ge 0 -\frac{\ell\cdot n}{2^{\nBits}} \ge -\lnrErrSymbol/\lg n
,\label{eq5}
\end{align}
where the last inequality is immediate from our choice of the number of bits -- $\nBits\triangleq\ceil{\lgp{n/\mu}+\lg\lg n}$.
We now split to cases based on the value of $\mu$. 
We
start with the simpler $\mu< 2\cdotpa{1-1/\lg n}$ case, in which $\nu=1$ (and consequently, $out\equiv 0$). 
This allows us to express the algorithm's error of~\eqref{eq4} as
\begin{align}
\widehat{S_i} - S_i = \mathfrak{r_0}  + \xi - \parentheses{\sensitivity - 1/2}.
\end{align}
We now use \eqref{eq:s0},\eqref{eq5}, and the definition of \sensitivity{} to obtain:
\begin{align*}
\widehat{S_i} - S_i = \mathfrak{r_0}  + \xi - \parentheses{\sensitivity - 1/2} \le -1/2.
\end{align*}
Similarly, we can bound it from below:
\begin{align*}
\widehat{S_i} - S_i = \mathfrak{r_0}  + \xi - \parentheses{\sensitivity - 1/2} \ge \xi - \parentheses{\sensitivity - 1/2} \ge -\lnrErrSymbol + 1/2.
\end{align*}
We established that if $\nu=1$ we achieve the desired approximation.
Henceforth, we focus on the case where $\mu\ge 2\cdotpa{1-1/\lg n}$, which means that $\nu=\floor{\mu\cdotpa{1-1/\lg n}} > 1$ and $\lastBit\in\set{0,1}$.
We now consider two cases, based on the value of \lastBit.
\begin{enumerate}
	\item {\textbf{$\bm{\lastBit=1}$ case.}\\}
	In this case, we know that after the processing of element $x_\nu$ the value of $\mathfrak{r}$ was at least $\sensitivity$ (Line~\ref{line:setBit}). This implies that $\mathfrak{r_0} + \sum_{d=1}^{\nu} Round_{\nBits}(x_d) \ge \sensitivity$ and equivalently
	\begin{align*}
	\mathfrak{r_0} + \sum_{d=1}^{\outBits} Round_{\nBits}(x_d) \ge \sensitivity - \sum_{d=\outBits+1}^{\nu} Round_{\nBits}(x_d).
	\end{align*}
	Substituting this in~\eqref{eq4}, and applying~\eqref{eq5}, we get that:
	\begin{align*}
	\widehat{S_i} - S_i &= \mathfrak{r_0} + \sum_{d=1}^{\outBits} Round_{\nBits}(x_d)  + \xi - \parentheses{\sensitivity - 1/2} -\ell\cdot\outBits\\
	&\ge \sensitivity - \sum_{d=\outBits+1}^{\nu} Round_{\nBits}(x_d)  + \xi - \parentheses{\sensitivity - 1/2} -\ell\cdot\outBits\\
	&\ge - \parentheses{\sum_{d=\outBits+1}^{\nu} \ell}  + \xi + 1/2 -\ell\cdot\outBits\\
	&\ge -\lnrErrSymbol/\lg n -\ell\floor{\mu\cdotpa{1-1/\lg n}}+1/2 \ge - \lnrErrSymbol + 1/2.
	\end{align*}
	In order to bound the error from above we use \eqref{eq:s0} and \eqref{eq5}:
	\begin{align*}
	\widehat{S_i} - S_i &= \mathfrak{r_0} + \sum_{d=1}^{\outBits} Round_{\nBits}(x_d)  + \xi - \parentheses{\sensitivity - 1/2} -\ell\cdot\outBits\\
	&\le \sensitivity - 1 + \ell\cdot \outBits  - \parentheses{\sensitivity - 1/2} -\ell\cdot\outBits \le -1/2.
	\end{align*}
	\item {\textbf{$\bm{\lastBit=0}$ case.}\\}	
	Here, since the value of \lastBit{}  is $0$, we have that $\mathfrak{r_0} + \sum_{d=1}^{\nu} Round_{\nBits}(x_d) < \sensitivity$ and thus
	\begin{align*}
	\mathfrak{r_0} + \sum_{d=1}^{\outBits} Round_{\nBits}(x_d) \le \sensitivity - \sum_{d=\outBits+1}^{\nu} Round_{\nBits}(x_d) - 1.
	\end{align*}
	We use this for the error expression of~\eqref{eq4} to get:
	\begin{align*}
	\widehat{S_i} - S_i &= \mathfrak{r_0} + \sum_{d=1}^{\outBits} Round_{\nBits}(x_d)  + \xi - \parentheses{\sensitivity - 1/2}\\
	&\le \sensitivity - \sum_{d=\outBits+1}^{\nu} Round_{\nBits}(x_d) - 1  + \xi - \parentheses{\sensitivity - 1/2} \le -1/2\\
	\end{align*}
	We now use \eqref{eq:s0}, \eqref{eq5}, and the fact that $\outBits\le \nu$ to bound the error from below as follows:
	\begin{align*}
	\widehat{S_i} - S_i &= \mathfrak{r_0} + \sum_{d=1}^{\outBits} Round_{\nBits}(x_d)  + \xi - \parentheses{\sensitivity - 1/2}\\
	&\ge \xi - \parentheses{\sensitivity - 1/2} \ge -\lnrErrSymbol + 1/2.
	\end{align*}
\end{enumerate}
Finally, we need to cover the case of $i\le o$. In this case, we can return $\mathfrak r-\parentheses{\sensitivity - 1/2}$ as the estimation. This directly follows from~\eqref{eq:s0} and the fact that within a chunk we simply sum the rounded values (Line~\ref{line:sum}).
We established that in all cases $-\lnrErrSymbol < \widehat{S_i} - S_i \le 0$.
\qedhere
\end{proof}
\fi
\end{document}